\documentclass[12pt]{article}

\usepackage{indentfirst,comment,color}
\usepackage{amsfonts}
\usepackage{amsthm}
\usepackage{amsmath}
\usepackage[norule,bottom]{footmisc}
\usepackage[justification=centering,textfont={sc},labelfont={rm}]{caption}
\usepackage{varioref}
\usepackage{chicago}
\usepackage{apalike}
\usepackage{setspace}
\usepackage[flushleft]{threeparttable}
\usepackage[online]{threeparttablex}
\usepackage{longtable}
\usepackage{rotating}
\usepackage[labelsep=space]{caption}
\usepackage[margin=1in]{geometry}
\usepackage{pdflscape} % for 'landscape' environment
\usepackage{booktabs}

\newtheorem{theorem}{Theorem}[section]

\newtheorem{proposition}[theorem]{Proposition}

\usepackage[colorlinks,citecolor=blue,urlcolor=blue,bookmarks=false,hypertexnames=true]{hyperref}

\begin{document}

\title{Countering Misinformation on Social Media Through Educational Interventions: Evidence from a Randomized Experiment in Pakistan}
\author{Ayesha Ali and Ihsan Ayyub Qazi\thanks{%
Ali: Department of Economics, Lahore University of Management Sciences (LUMS), DHA, Lahore Cantt, 54792, Lahore, Pakistan, ayeshaali@lums.edu.pk. Qazi: Lahore University of Management Sciences (LUMS), DHA, Lahore Cantt, 54792, Lahore, Pakistan, ihsan.qazi@lums.edu.pk. We thank Ghazala Mansuri, David G. Rand, Asim Khwaja, Alex Leavitt, Andrew Guess, Zafar Ayyub Qazi, Dina Tasneem, Kate Vyborny, Syed Zahid Ali, Syed Ali  Hasanain,  Naved  Hamid,  David  Yang  and  participants  at  MIT  Sloan  School  of  Management, Facebook Research, American University of Sharjah, Lahore School of Economics,and Lahore University of Management Sciences research seminars, for valuable comments. We gratefully acknowledge support by the Facebook Integrity Foundational Research Award 2019 from Facebook Research in the form of an unrestricted gift and the Lahore University of Management Sciences (LUMS) Faculty Initiative Fund grant. We thank Sara Obaid, Noor ul Islam and Noor Alam for project coordination and research assistance, and Ahsan Tariq and Mohammad Mallick at the Institute of Development Alter-natives (IDEAS) Survey Wing for leading field implementation and data collection.  We also thank Zartash Uzmi for appearing in our intervention video as a presenter. The study was approved by the IRB at Lahore University of Management Sciences (Protocol Number: LUMS IRB/03252019). This study is registered in the AEA RCT Registry with the unique identifying number: ``AEARCTR-0004003".}}

\maketitle

\begin{abstract}
Fake news is a growing problem in developing countries with potentially far-reaching consequences. We conduct a randomized experiment in urban Pakistan to evaluate the effectiveness of two educational interventions to counter misinformation among low-digital literacy populations. We do not find a significant effect of video-based general educational messages about misinformation. However, when such messages are augmented with personalized feedback based on individuals’ past engagement with fake news, we find an improvement of 0.14 standard deviations in identifying fake news. We also find negative but insignificant effects on identifying true news, driven by female respondents. Our results suggest that educational interventions can enable information discernment but their effectiveness critically depends on how well their features and delivery are customized for the population of interest.


\end{abstract}

\section{Introduction}
Misinformation is a growing concern in developing countries where low-cost smartphones and mobile Internet access has led to increasing use of social media platforms as sources of information and a place for social and political activity.\footnote{According to the `Digital in 2020' report \cite{digital_2020}, social media use in developing countries is increasing with annual growth rates ranging from 12\% to 38\%.} This trend has brought many new users online, including those with limited digital literacy, which is the ability to access, understand and use information online \cite{survey_measure}.
According to the International Telecommunications Union \cite{itu_skills}, less than half of the population in 40 developing and emerging countries, possesses basic computer skills such as copying a file or sending an email with an attachment.
These relatively inexperienced technology users are joining social media in a landscape in which misinformation is a common phenomenon.
This can have far-reaching consequences for individuals and society.\footnote{These effects can range from  election interference (\citeNP{misinfo_election}; \shortciteNP{grinberg}), polarization (\shortciteNP{affective}; \citeNP{cross_polar}), lack of adherence to preventive measures during the COVID-19 pandemic \shortcite{polarization_health,misinfo_covid}, and to even conflict (\shortciteNP{conflict_outcomes}; \citeNP{conflict}; \citeNP{social_hate}; \citeNP{capitalhill}).}

Can simple educational interventions improve the ability to identify misinformation among low digital literacy populations?
% to what an extend "design" matters -- emphasize
Educating users about misinformation and the digital space can offer a promising pathway to reduce the prevalence and impact of fake news.\footnote{While platform providers, such as Facebook, can filter and flag potentially fake content, it is challenging to remove it completely due to the nuanced nature of misinformation that varies based on context, demographics, and culture (\citeNP{rand_crowd}; \citeNP{factcheck2}; \shortciteNP{factcheck3}; \shortciteNP{factcheck}). In messaging services, such as WhatsApp, that provide end-to-end encryption of messages, directly filtering or flagging content may not even be possible.}
Although the benefits of educating people on countering misinformation seem plausible, there is little direct evidence about the effectiveness of educational interventions to counter misinformation in developing countries.

In this paper, we evaluate the effectiveness of two educational interventions for countering misinformation using a clustered randomized control design with 750 participants drawn from urban Pakistan. A cluster in our study is a small geographic area of 100 m x 100 m, covering low and middle income areas in the city of Lahore. To avoid spillovers, we assign 50 clusters each to control and the two treatment arms. Within each cluster, we draw five individuals from five distinct households. A unique feature of our study is that we recruit our sample from the field, instead of relying on online survey platforms that are less representative of low digital literacy users.\footnote{For example, \citeN{guess_dl} and \citeN{AMT} found that MTurk samples are not representative of low digital literacy users and they differ not just demographically but also in their online behaviors.} 
In the baseline, we measure the extent to which individuals are able to correctly identify fake news using a list of actual true and fake news stories circulated on social media.
This is followed by the delivery of educational interventions in the treatment arms.
In the first treatment (treatment 1), individuals are shown an informational video in the national language that educates them about common features of misinformation. In the second treatment (treatment 2), individuals are first shown the video and then given personalized feedback about their own responses to fake news stories shown at baseline. The feedback highlights the features of each fake news item that can enable the individuals to identify the news as fake.

We measure the effectiveness of our interventions through two endline surveys, by asking individuals about their beliefs regarding two additional sets of news stories circulated on social media. We also measure different ways users engage with the news including why they label a news item as true or fake, the degree of emotional arousal, intent to verify and share on social media. The endlines surveys were conducted after one week and then 4 to 6 weeks after the treatment was delivered.
We were unable to directly observe online behaviors due to the high usage of WhatsApp an end-to-end encrypted platform (96\%) and low usage of open platforms like Twitter (11\%) in our sample. Tracking individuals on Facebook (46\% usage in our sample) is also difficult due to privacy and safety concerns especially among female users \cite{safe_digital}. To assess the longevity and external validity of our treatments, we conducted a follow-up phone survey fifteen months after the interventions on assessing the accuracy of beliefs about news related to the COVID-19 pandemic as well as information-seeking and preventative behaviors during the pandemic.

Our results show that treatment 1 does not have any significant impact. However, we find that those who received treatment 2 were 0.14 standard deviations more likely to correctly identify fake news relative to the control group. There is no significant effect on correctly identifying true news. As a result, the overall effect of treatment 2 on all news is 0.076 standard deviations, which points to improved discernment of news. Those who received treatment 2 are 0.13 standard deviations less likely to report non-neutral (positive or negative) emotional reactions, relative to the control group.
%which suggests that treatment 2 had a mitigating effect on initial reactions to news.
Furthermore, individuals who received treatment 2 are less likely to mention prior beliefs and more likely to mention features of the news, such as unreliable source, as the reason why they consider a news story to be fake or true. 

We show through various robustness checks that experimenter demand effects are unlikely to explain the effect of treatment 2. For example, we do not find any effect on placebo news items (constructed fake news but with some features of true news) which suggests that respondents are not simply terming news as fake to please the experimenter. We are also able to reject the explanation that the treatment effect comes from randomly guessing more news items as fake. Furthermore, we do not find evidence of a larger impact in a sub-group who might identify news as fake in order to conform to the socially desirable trait of not falling for fake news. 

More detailed results show that individuals respond heterogeneously to treatment 2. In particular, the effect of treatment 2 is driven by those with higher digital literacy and by the male respondents in our sample.
Moreover, our results also show that there is limited impact on female respondents. We find that they are likely to label both fake and true news as fake which suggests no improvement in information discernment and increased skepticism even towards true news on social media. 

Our phone survey on COVID-19 misinformation, which took place after the first wave of the pandemic in Pakistan, shows a high level of general awareness and accuracy of beliefs about COVID-19 related fake news in the entire sample. We find that general awareness masks any potential treatment effect. However, we find a 0.09 standard deviation improvement in the accuracy of beliefs in the treated individuals who have never seen or heard the news before. We also find evidence of more conservative information-seeking and sharing behaviors on social media and greater adoption of preventative health behaviors such as social distancing and mask ownership among the treated group relative to the control group. We interpret these findings to be indicative of long-lasting benefits of educational interventions.

To explain our findings, we present a simple conceptual framework of how individuals form beliefs about the accuracy of news encountered on social media. We model individuals as motivated reasoners, i.e., they like news conforming to their prior beliefs.\footnote{There is a large body of literature that models individuals as seeking to confirm prior beliefs to minimize psychological costs in the tradition of \citeN{socialpsych}, for example \cite{market_news}, \cite{model2}. Empirical evidence on misinformation also shows that people are more likely to believe and share fake news that conforms to their prior beliefs \cite{misinfo_election,share}.} However, they are also able to apply analytical reasoning to form beliefs about the news, rather than relying on ``gut feelings" or intuitions.\footnote{According to dual-process theories, human cognition can be characterized by a distinction between intuitive, autonomous (System 1) processes and deliberative, analytic (System 2) processes \cite{kahneman,dual_evan,3_stage}. In the misinformation literature, \citeN{rand_lazy} show that analytical reasoning improves accuracy of beliefs about the news.} In our model, educational interventions decrease the cost of applying analytical reasoning by providing simple decision rules to evaluate the credibility of sources as well as indicators of problematic content without expending significant effort. Thus, as a result of educational interventions, individuals will be more likely to form correct beliefs about fake news, even if it conforms to their prior beliefs. 

Our work is related to emerging literature in economics, psychology, political science, and related fields that seeks to understand the effectiveness of ways to reduce the impact and distribution of misinformation. Most of the literature has examined the effectiveness of behavioral interventions such as nudges, which are coupled with the presented news. For example, attaching warnings to news stories disputed by third-party fact-checkers \cite{rand_implied} and using crowd-sourcing to generate trust ratings to differentiate reliability of sources have been found effective in increasing the accuracy of beliefs about misinformation \cite{rand_crowd}. Moreover, a general warning about misleading information on social media or attaching ``Disputed'' or ``Rated false'' tags to specific headlines was also found to be effective in reducing perceived accuracy of false headlines \shortcite{tags}. However, an unintended consequence of these approaches is an increase in skepticism towards true news.

Despite this emerging literature, there is little evidence on the effectiveness of educational interventions that focus on teaching social media users how to recognize fake news, without nudges or attaching labels to the news itself, and especially in the context of low digital literacy users. We are aware of two contemporaneous media and information literacy evaluation programs carried out in developing countries. \citeN{badrinathan} evaluates the effectiveness of in-person media literacy training on two tools for identifying misinformation: navigating a fact-checking website and reverse image searching using a low digital literacy sample from India. The paper finds no significant effect on the respondents' ability to identify misinformation. Similarly, \shortciteN{guess_experiment} evaluated the effectiveness of a social media literacy intervention in which individuals in United States and India were provided with tips to spot fake news. Their India sample consisted of high digital literacy users recruited online and a rural sample with very low familiarity with social media. The authors find a positive effect of giving tips online to highly educated social media users which declines over time. There is no effect in their rural sample for whom the tips are given verbally by enumerators.

Our paper adds to the literature on countering misinformation in developing countries among individuals with low digital literacy. %\textcolor{red}{We find evidence in support of educational interventions which is broadly consistent with prior work in this area.} 
Moreover, by evaluating two treatments we also provide evidence on which features of such programs are likely to be important in improving the ability to identify fake news. Our results show that despite their popularity, general educational messages do not produce any significant impact. Rather educational programs involving personalized feedback and guidance are more effective and longer-lasting for an average social media user in populations where digital literacy is lower relative to developed countries. Our findings also highlight the importance of customizing interventions from the gender perspective, as we show that females in our sample do not experience improved discernment between fake and true news. While we are not able to directly observe the online or offline behaviors of these respondents, the results from our phone survey on COVID-19 is suggestive of positive long-term effects on real-world behaviors in the context of the pandemic.

The remainder of the paper is structured as follows: section \ref{sec:context} outlines the context and describes the experiment, data, and empirical methodology; section \ref{sec:results} presents the results; section \ref{sec:model} presents a model to explain belief formation and impact of educational interventions; section \ref{sec:discussion} discusses our results and their implications; and section \ref{sec:concl} concludes.

\section{Context and Experiment Design}
\label{sec:context}
We choose to focus our study on social media users in Pakistan, a large developing country where approximately 18\% of the population or 37 million people use social media. According to the `Digital 2019: Pakistan' report \cite{portalpakistan}, social media penetration has grown from just 4\% in 2013 to 18\% in 2019. Many of these users are newly connected to the Internet through their smartphones. Facebook.com and whatsapp.com are among the websites with the largest average monthly traffic in the country \cite{topwebsites}. Furthermore,  opinion polls conducted with a nationally representative sample of Internet users show that 40\% visit Facebook and YouTube to get news \cite{gallup2019}.

Our sample is drawn from low and middle income areas of the city of Lahore; the second largest city of Pakistan with a population of around 11 million. We focus on low and middle income areas as we expect that digital literacy will be lower for these strata of the urban population.
We recruit our sample from the field, instead of relying on online survey platforms that are less representative of low digital literacy users \cite{guess_dl,AMT}.
We examine the effectiveness of two educational interventions using a cluster randomized design, where individuals within a small geographic cluster (of 100 m x 100 m) are assigned to either control, treatment 1 or treatment 2 status. In order to minimize the potential of spillovers within this small area, we carry out randomization at the cluster rather than the individual level. Our final sample consists of 150 geographic clusters, with 5 individuals per cluster, summing to a total of 750 individuals divided equally into control, treatment 1, or treatment 2. Online Appendix A contains details of the sampling, survey procedures, and randomization.

\subsection{Baseline}
In the first visit, we carry out the baseline survey which includes demographic and social media use questions, followed by questions to measure digital literacy. We capture digital literacy using two types of questions (i) basic skills such as the ability to use social media without assistance, connect to WiFi and/or mobile data, use Google search and read English on social media, and (ii) specific skills related to social media applications measured by asking whether the individual is aware of and uses common application features. We ask about sixteen common WhatsApp and Facebook features (eight for each application).\footnote{Facebook features are as follows: create post, like, share, comment, privacy settings, report user, hide, and sponsored post. WhatsApp features are view chat, reply, record audio message, forward message, delete message, report user, block user, and message seen.} These questions are consistent with prior work, which validated survey measures of digital literacy through participant observation, and found that composite variables of survey questions that measure users’ knowledge of computer- and Internet-related terms and functions are good predictors of actual digital literacy \cite{survey_measure,update_survey_measure}.

After gathering these baseline characteristics, we measure beliefs and perceptions about news on social media. To do so, we ask a series of questions about a set of actual true and fake news stories drawn from social media. All individuals view the same set of stories, however, we randomize the order in which they are shown. In the baseline we show three true and fake news stories each, covering current events and general interest topics. The news stories are shown in the form of screenshots of messages, posts, or tweets, similar to how users would typically receive news on social media. Some of the news stories are in English only, while others are completely or partially in the national language, Urdu. It is fairly common for people to receive news in English as a large fraction of content on social media is in English. In order to make sure that the respondents understand the content, they are first asked to carefully view a printed version of the screenshot of the news and then asked to listen to an audio recording of the news. If the news was in English the audio translates it into Urdu. After viewing and hearing the news we ask a series of questions to ask whether they recalled seeing the news before, their beliefs about whether the news is true or not, their reactions, and their level of engagement with the news. Our main outcome of interest is the response to the question, ``Right now do you think the news is true or not?" We use this question to create our primary outcome, which is an indicator equal to 1 if they correctly identify news, either fake or true.

In addition to the fake and true news stories, we also show one placebo (or constructed) news story. This story is similar in content and format to our list of actual fake and true news stories, but was not drawn from social media news. The responses on the baseline placebo news item are used to measure and control for false recall, that is, if an individual is more or less likely to say that they recall seeing the news before, just because they are part of a study or due to any design feature of our study. False recall might also reflect pure recall error on the part of the individual which would happen even outside the study setting. At the end of the baseline survey, individuals in the treatment arms receive one of the educational interventions described below. The baseline survey and delivery of treatments in the treatment arms is carried out by a pair of male and female enumerators. 

\subsection{Educational Interventions}
After the baseline survey, individuals in the first treatment group watch a three-minute video about fake news in Urdu language. The video explains briefly and simply what is fake news, why it is generated, and the negative consequences of fake news. The remaining part of the video is dedicated to explaining the following three tips to identify fake news: (i) untrustworthy or missing source, for example, unknown author, unverified account, or absence of an authentic link supporting the news, (ii) poor quality of news, that can be spotted visually, for example, altered images or videos, informal or incorrect language, or excessive use of hashtags and emojis, and (iii) content or language that reflects bias, and tries to provoke extreme reactions and emotions. These tips are consistent with the Facebook’s ``Tips to Spot False News", which were developed in collaboration with the nonprofit First Draft and subsequently promoted at the top of users’ news feeds in 14 countries in April 2017. A variant of these tips was later distributed by WhatsApp in newspaper advertisements in Pakistan and India in 2018 \cite{fb_tips_v2,guess_experiment}.

Individuals in the second treatment group watch the video and then receive personalized feedback based on how successful they were in identifying each of the three fake news items shown at baseline. The personalized feedback was delivered verbally by the enumerators using pre-designed and printed feedback cards. The feedback consisted of two components: (i) a score of 0 or 1, based on whether the individual correctly identified the fake news item as fake and (ii) drawing attention to the features of the fake news item, that is source, quality, or bias, which could be used to identify the news as fake. The feedback cards were designed to be encouraging and informational, ensuring that the individuals understand why they got something right or wrong. The feedback was delivered regardless of whether they identified the fake news item to be fake or true. 
\subsection{Endline and COVID-19 Follow-up Survey}
We conducted our baseline survey and the two endline surveys between May 2019 to July 2019. The first endline survey comprised six fake and three true news items, which were shown to the participant in the same manner as in the baseline.  For each news item, we ask the same set of questions as in the baseline. In the first endline, we also include two placebo news items to assess experimenter demand effects (or Hawthorne effects) of our treatments (i.e., do the treatments make individuals act in a manner that they perceive is desired by the experimenter, which in this case is identifying news as fake). The second endline was carried out 4 to 6 weeks after the first endline. It consisted of five fake and three true news stories presented in a similar manner as the baseline and first endline. 

We were able to reach 731 of the RCT participants in the first endline (response rate of 97.5\%). Online Appendix Table \ref{table:survey_response_rates} shows that there are no systematic differences in the attrition rate across the three groups. 92\% of the first endline surveys were completed by the fifth day after the baseline and the remaining were completed by the ninth day.  We were able to reach 744 RCT participants (response rate of 99.2\%). Overall, we find that attrition was random and not correlated with the treatment status.

In addition to these two endline surveys that were conducted in the same format as the baseline, we also conducted a follow-up phone survey in September 2020. We were able to successfully reach 82\% of the participants who were part of our RCT. The response rate was balanced across the three groups (online Appendix Table \ref{table:survey_response_rates}). The phone survey exclusively focused on news statements related to the COVID-19 pandemic. In this survey, we included seven fake news items and three true news items. For each news item we asked two questions (i) whether the individual had heard the news before, and (ii) to what extent would they agree with the news statement (ranging from agree, disagree to not sure). In this survey, our goal was to assess the extent to which individuals were affected by COVID-19 related misinformation, and whether our treatments had any effect in making them less susceptible to fake news.

\subsection{Empirical Strategy}
In order to find the average treatment effect of each intervention on the probability of correctly identifying a news item as true or fake, we estimate the following equation:
\begin{equation}
    Y_{ik} = \alpha_0 + \alpha_1 X_{i}+ \beta_1 T_{1i} + \beta_2 T_{2i} +  \epsilon_{ik}
\end{equation}
where $Y_{ik}$ is the outcome of individual $i$ for news items shown post treatment, $X_{i}$ is a vector of baseline controls (described later). $T_{1i}$ is an indicator for treatment 1 and $T_{2i}$ is an indicator for treatment 2. $\beta_1$ and $\beta_2$ are the average treatment effects for treatment 1 and treatment 2, respectively.
The outcome variable is equal to 1 if the individual correctly identifies fake news as fake or true news as true, and 0 otherwise.
%We term this as the ability to correctly identify news.
We standardize this variable by using the mean and standard deviation of the control group so that treatment effect is estimated in terms of standard deviation changes.

We also examine the effect of our interventions on how individuals engage with the news items. Specifically, we measure the effect on reporting extreme emotions. This variable is zero if the respondent reports feeling neutral, 1 if they report feeling positive or negative and 2 if they report feeling extremely positive or negative. We also study the effect on each of the following binary variables: whether the individuals say they will discuss the news with friends or family, search online about the news, and share news on social media, or do nothing. Since our interventions are designed to make individuals aware of common features of fake news, we also evaluate the effect on the reasons mentioned by the respondents for why they believe a news item to be true or fake. Specifically, we construct binary variables for whether an individual mentions prior political or religious beliefs, source problem, quality problem, or biased content, as reasons for why they consider a news item to be true or fake. We also standardize these outcomes by the mean and standard deviation of the control group. In addition, we examine the existence of heterogeneous treatment effects by the following baseline moderators: digital literacy, gender, age, and education.

\subsection{Sample Characteristics and Balance Test}
We now describe some important characteristics of our sample. In terms of household characteristics, we find that the median household monthly expenditure was Rs. 30000 (approximately USD 200).\footnote{According to the Household Integrated Economic Survey 2018-19 conducted by the Pakistan Bureau of Statistics, the average monthly household expenditure in urban areas was Rs. 47362. The mean expenditures by quintile were Rs. 23515, Rs. 29130, Rs. 32931, Rs. 38689, and Rs. 64681 \cite{income}.} Households have 6 to 7 members with 2 to 3 members who are social media users.\footnote{We define such a user as someone of age 18 or above having a WhatsApp or Facebook account.} Roughly half of our sample is female, the mean age is 29 years and the most common educational qualification is Intermediate (Grade 12).

In terms of social media use, 96\% of our sample has a WhatsApp account, 46\% has a Facebook account, 11\% has a Twitter account, and spend approximately 3 hours per day on social media.
The popularity of end-to-end encrypted messaging services like WhatsApp and low penetration of open platforms (e.g., Twitter) in our sample makes it challenging to directly observe online behaviors.
Tracking individuals on Facebook is also difficult due to privacy and safety concerns especially among female users \cite{safe_digital}.
%Moreover, we find a significant gender gap in the use of Facebook.
In addition to using social media to view content, about two-thirds say they use social media to share content and 40\% say they use social media to create content.

When looking at digital literacy, we find that approximately 85\% report being able to use social media without assistance,  83\% say they can connect to WiFi or mobile data, followed by 78\% reporting the ability to use Google search, and 63\% report being able to read English text on social media.
In our empirical work, we take the sum of these four binary indicators (standardized by the mean and standard deviation of the control group) as our measure of basic digital literacy skills, referred to as the basic digital literacy score.
In addition, we also create indices using data on knowledge and usage of common WhatsApp and Facebook features. While the basic digital literacy score is effective at separating those at the low end and middle of the digital literacy spectrum in our sample, the latter indices are better at separating individuals at the high end. This is consistent with prior work that recommends having different measures for capturing variations at the low end and the high end of digital literacy \cite{guess_dl}. For our empirical work, we use principal component analysis to define separate scores for knowledge and usage of WhatsApp and Facebook features, standardizing these measures using the mean and standard deviation of the control group. We refer to these as the WhatsApp Score and Facebook Score.

In online Appendix Table \ref{balance}, columns 1 to 3, we report the summary statistics of our RCT sample by treatment status. In columns 4 to 6, we report the p-value for the pairwise difference in means across the groups, that is treatment 1 versus control (column 4) and treatment 2 versus control (column 5), and treatment 1 versus treatment 2 (column 6). Standard errors are clustered at the sample grid level. Of the twenty-four baseline characteristics reported, we find significant differences across the groups for only four characteristics. Participants from the control group have a greater likelihood of owning a smartphone (relative to both treatment groups). They are also more likely to report knowing more WhatsApp features (relative to both treatment groups). Participants in the second treatment group are more likely to falsely recall placebo news shown at baseline (relative to the control group), and likely to use fewer Facebook features (relative to the first treatment group). All other characteristics are not significantly different across the three groups.

Therefore, we conclude that our randomization worked and that the three groups are similar in terms of baseline socio-economic, demographic, digital literacy, and social media use variables.

\section{Results}
\label{sec:results}
\subsection{Treatment Effect}
Table \ref{table:ate} shows the average treatment effect of our interventions on correctly identifying news. While we do not find any significant effect of treatment 1, treatment 2 improves the probability of correctly identifying fake news by 0.14 standard deviations relative to the control group (statistically significant at the 5 percent level) when the endlines are pooled.
At endline 1 and endline 2, the effects were 0.12 and 0.15 standard deviations, respectively. The effect is larger but less precisely estimated in the second endline. The precision is improved by adding baseline controls to the regressions. In Section \ref{sec:alternate} we perform various checks to examine the extent to which these effects are driven by experimenter demand effects, that is respondents are labeling news as fake since as they are trying to please the experimenter. 

\begin{table}[ht]
\begin{center}
\begin{threeparttable}
\footnotesize

  \centering
% \captionsetup{font=normalsize}
\captionsetup{font=footnotesize} 
\caption{Average Treatment Effect}
 \label{table:ate}
    \begin{tabular}{lcccccc}
    \toprule
    \multicolumn{7}{l}{Dependent variable: Correctly identify news}\\
    \hline
    &\multicolumn{2}{c}{Fake}&\multicolumn{2}{c}{True }&\multicolumn{2}{c}{All }\\
    &(1)&(2)&(3)&(4)&(5)&(6)\\
    \toprule
\multicolumn{7}{l}{\textit{Panel A - Pooled Effect}}\\
\\
Treatment 1& -0.022&-0.020&-0.005&-0.022&-0.016&-0.019\\
&(0.065)&(0.062)&(0.041)&(0.041)&(0.048)&(0.047)\\
Treatment 2&0.136**&0.146**&-0.048&-0.052&0.071&0.076*\\
&(0.061)&(0.059)&(0.040)&(0.039)&(0.045)&(0.044)\\
&&&&&&\\
Observations &   \multicolumn{2}{c}{8,111} &   \multicolumn{2}{c}{4,428} & \multicolumn{2}{c}{12,539}\\   
\hline
\multicolumn{7}{l}{\textit{Panel B - Endline 1}}\\
\\
Treatment 1&	-0.008&	-0.001&	-0.043&	-0.045&-0.019&	-0.024\\
&	(0.061)	&(0.058)&	(0.066)&	(0.063)&	(0.053)	&(0.052)\\

Treatment 2&	0.121**&	0.128** &	-0.034&	-0.038&0.069&	0.072\\
&	(0.060)	&(0.056)&	(0.061)&	(0.061)&	(0.051)	&(0.049)\\
&&&&&&\\
Observations  &   \multicolumn{2}{c}{4,386} &   \multicolumn{2}{c}{2,193}& \multicolumn{2}{c}{6,579}\\   
\hline

\multicolumn{7}{l}{\textit{Panel C - Endline 2}}\\
\\
Treatment 1&	-0.035&	-0.032&	0.036&	0.019&-0.009&	-0.013\\
&	(0.097)	&(0.094)&	(0.052)&	(0.053)&	(0.066)&	(0.065)\\

Treatment 2&	0.154&	0.167*&	-0.060&	-0.063&0.074&	0.081\\
&	(0.095)	&(0.092)&	(0.045)&	(0.046)&	(0.063)	&(0.061)\\
\\
Observations  &   \multicolumn{2}{c}{3,725} &   \multicolumn{2}{c}{2,235}& \multicolumn{2}{c}{5,960}\\   
\bottomrule
    \end{tabular}   
 \begin{tablenotes}
   \footnotesize 
   \item Notes: Standard errors clustered by sample grid in parentheses. * p$<$0.1,  ** p$<$0.05,*** p$<$0.01. Columns 1, 3, and 5 include no controls. Controls in columns 2, 4, and 6 are household monthly expenditures, fraction of social media users in household, baseline score, baseline false recall, gender, age, education, hours on social media, smart phone ownership, social media as primary news source, four basic digital literacy indicators, indicators for view, share and create content, and indicator for second endline in pooled regressions.
   \end{tablenotes}
\end{threeparttable}
\end{center}
  \label{tab:ate}
\end{table}

Columns 3 and 4 in Table \ref{table:ate} show the average treatment effect on the sample of true news items only. We do not find any significant effect of either treatment on correctly identify true news. One plausible explanation for no effect on true news could be that individuals were generally aware of the true news in both the treatment and control groups and that explains why we do not find any significant effect of the treatment. However, across both the treatment and control groups, the majority of individuals do not recall seeing the true news items before (65\% in the control group and 70\% in the treatment 2 group) suggesting general awareness about true news does not drive our results.
However, the negative coefficient of the treatment variable for true news (although insignificant) might be due to experimenter demand effects or increased skepticism towards social media news. We explore these explanations of the results in Section \ref{sec:alternate} and \ref{sec:discern}.

In columns 5 and 6 we report the effect on all news items which captures the overall effect of the treatment in improving discernment between true news and fake news. This can be regarded as a stronger test of whether our treatment causes individuals to be more discerning social media users. In the results with both endlines pooled together and baseline controls, we find a positive effect of 0.076 standard deviations (significant effect at 10\%) of treatment 2. This effect is dominated by improvements in correctly identifying fake news, which points to improved discernment. The results are similar in magnitude (although less precisely estimated) without controls and separately across the endlines.

The effects of treatment 2 on identifying fake news found in Table 1, compare favorably with previous literature on similar educational interventions. For example, \shortcite{guess_experiment} find an effect size of 0.11 in their online convenience sample in India and no effect on a rural sample that received tips to spot fake news in person. Similarly, \citeN{badrinathan} does not find any significant effect. The effect in this study is also larger than the effect size of 0.08 found in \shortcite{tags}, as a result of giving general warnings about misinformation on social media users in the US recruited from MTurk. 

\begin{table}[t]
\centering
\begin{threeparttable}
\footnotesize
\captionsetup{font=footnotesize}
\caption{Engagement with News}
 \label{table:engage_news}
\begin{tabular}{l llp{2cm}ll}

\toprule
Dependent variable:& Emotional&Do nothing&Discuss with friends/family & Search online & Share\\
&(1)& (2) & (3) & (4)& (5)\\
\toprule

Treatment 1  &	-0.015 &-0.048 &0.033&-0.001&	0.045\\
&(0.062)&(0.089)&(0.081)&(0.045)&(0.059)\\
P-value& 0.808&0.589&0.685&0.979&0.445\\
Q-value& 1.000&1.000&1.000&1.000&1.000\\
\\
Treatment 2 &	-0.140**&	0.148*&	-0.165**	&-0.024&	-0.051\\
&	(0.059)	&(0.081)&	(0.074)&	(0.042)	&(0.055)\\
P-value& 0.019&0.070&0.026&0.569&0.352\\
Q-value& 0.070&0.076&0.070&0.295&0.214\\
\bottomrule
\end{tabular}
\begin{tablenotes}
      \footnotesize
      \item Notes: Sample is all fake and true news items shown post-treatment in the two endlines. All regressions include baseline controls and endline wave dummy. Standard errors clustered by sample grid in parentheses. * p$<$0.1,  ** p$<$0.05,*** p$<$0.01.  P-value is the unadjusted p-value and q-value is the sharpened False Discovery Rate-adjusted two-stage q-value.
\end{tablenotes}
\end{threeparttable}
\end{table}

Table \ref{table:engage_news} shows the effect of the treatments on secondary outcomes measuring engagement with the news items. Those who received treatment 2 are 0.13 standard deviations less likely to report non-neutral (positive or negative) emotional reactions (statistically significant at the 5 percent level), relative to the control group.
This suggests that treatment 2 may have had a mitigating effect on initial reactions to news due to increased awareness about the nature of fabricated news on social media. They are also 0.14 standard deviations more likely to say that they will do nothing after hearing the news (significant at the 10 percent level), 0.16 standard deviations less likely to say that they will discuss the news with friends and family (significant at the 5 percent level). We do not find any effect on searching online for the news or intent of sharing it on social media. Table \ref{table:engage_news} also includes the sharpened q-values using the two-stage procedure described in \cite{anderson}, to adjust p-values for multiple hypothesis testing. All outcomes that have p-values below 0.05 or between 0.05 and 0.1 have sharpened q-values below 0.10. In additional results available upon request, we find that there are similar effects for each of these outcome variables when looking at fake and true news separately. This suggests that after receiving the treatment, individuals are less likely to feel strong emotions, less likely to discuss the news with others, but there is no change in behavior related to verification or sharing of the news.
\begin{table}[t]
\centering
\begin{threeparttable}
\footnotesize
\captionsetup{font=footnotesize}
\caption{Reasons for Why News is True/Fake }
 \label{table:reasons_true_fake}
\begin{tabular}{l llll}

\toprule
Dependent variable:& Prior Belief&Source&Quality & Bias\\
&(1) & (2) & (3) & (4)\\
\toprule
Treatment 1&-0.082&	0.065&	-0.026&	-0.016\\
&	(0.053)	&(0.069)&	(0.065)&	(0.046)\\
P-value& 0.124&0.343&0.696&0.726\\
Q-value& 0.985&1.000&1.000&1.000\\
\\
Treatment 2&-0.112**&	0.142*&	-0.047&	-0.010\\
&	(0.049)	&(0.075)&	(0.065)	&(0.041)\\
P-value& 0.024&0.061&0.47&0.816\\
Q-value& 0.107&0.107&0.457&0.690\\
\bottomrule
\end{tabular}
\begin{tablenotes}
      \footnotesize
      \item Notes: Sample is all fake and true news items shown post-treatment in the two endlines. All regressions include baseline controls and endline wave dummy. Standard errors clustered by sample grid in parentheses. * p$<$0.1,  ** p$<$0.05,*** p$<$0.01.  P-value is the unadjusted p-value and q-value is the sharpened False Discovery Rate-adjusted two-stage q-value.
     
\end{tablenotes}
\end{threeparttable}
%\vspace{-0.1in}
\end{table}

Since our educational interventions focused on teaching individuals specific ways to identify fake news, in Table \ref{table:reasons_true_fake}, we estimate the effect of the treatments on the reasons stated for labeling news as true or false. We find that individuals are 0.11 standard deviation less likely to mention prior beliefs (statistically significant at the 5 percent level) and 0.14 standard deviation more likely to mention a source problem as the reason for their answer (statistically significant at the 10 percent level). We also find that for outcomes with p-values less than 0.1, the sharpened q-values were 0.107. This suggests that treatment 2 helps individuals in engaging more analytically with news (e.g., by helping them examine the characteristics of the presented news such as its source). If we separate the effect by fake and true news, we find that the effect on prior beliefs and source are driven by fake news items, which further suggests that individuals are learning and applying what was taught in the second intervention.

\subsection{Experimenter Demand Effects}
\label{sec:alternate}
In this section, we examine the extent to which the effect of treatment 2 might be driven by experimenter demand effects. Since our intervention explains to participants the negative consequences of fake news, it might leave an impression that the researchers are expecting participants to identify news items as fake. Thus, in the endline surveys, participants are likely to identify news as fake, in order to satisfy the researchers. However, outside the experimental setting, individual behavior may not necessarily change. Such changes in behavior among the treatment group are known as experimenter demand effects (or Hawthorne Effects) \cite{duflo_nber_tool,ede}. We examine the plausibility of experimenter demand effects in explaining the observed treatment effect. 

%\clearpage
\begin{table}[t]
\centering
\begin{threeparttable}
\footnotesize
%\hline
\captionsetup{font=footnotesize}
\caption{Treatment Effect by Social Desirability Trait }
\label{table:effect_by_social}
\begin{tabular}{l p{1.5cm} p{1.5cm}p{1.5cm}}
\toprule
\multicolumn{4}{l}{Dependent variable: Correctly identify news}\\
\hline
 & Fake & True &  All \\
&(1) & (2) & (3)\\
\toprule
Treatment 1 & -0.003 & 0.004 & -0.001	\\
 & (0.073) & (0.057) & (0.059) \\
 Treatment 1 x Social Desirability & -0.035 & -0.058 & -0.043\\
& (0.096) & (0.082) & (0.077)\\

Treatment 2	 & 0.141* & -0.044 & 0.076  \\
& (0.074) & (0.059) & (0.059)\\
Treatment 2 x Social Desirability  &  0.011 & -0.018 & 0.001\\
& (0.095) & (0.085) & (0.076)  \\ 

\bottomrule
\end{tabular}
\begin{tablenotes}
    \footnotesize
      \item Notes: Sample is all news items shown in two endline surveys. All regressions include baseline controls and endline dummy. Standard errors clustered by sample grid in parentheses.* p$<$0.1,  ** p$<$0.05,*** p$<$0.01.
\end{tablenotes}
\end{threeparttable}
\end{table}

Our first robustness check involves adopting a modified version of the approach used in \cite{ede_check}, to zoom into a sub-sample of our experimental subjects in which social desirability trait may be  dominant. Specifically using our baseline survey data, we identify individuals who answered `no' when asked if they have ever fallen for fake news. If the participants sense that the experiment is about correctly identifying fake news, and they would like to please the experimenter, they might be more likely to respond `no' to this question. These individuals should also make a greater effort to identify news as fake in the endline, and the treatment effect should be driven by this sub-sample. 44\% of our respondents fall into this category. Table \ref{table:effect_by_social} shows that the additional treatment effect for this sub-sample is small and insignificant as compared to the treatment effect on the remaining sample for fake news items. For true news we would expect a larger negative coefficient for the social desirability sample if respondents were identifying true news as fake due to experimenter demand effects. We do not observe any evidence for such effects. Overall, we do not find any evidence of a significant differential response to the treatment in the social desirability sample when pooling together all news items.  

\begin{table}[tbp]
\centering
\begin{threeparttable}
\footnotesize
%\hline
\captionsetup{font=footnotesize}
\caption{Falsification Test with Different Rates of labeling News as Fake}
\label{table:falsification1}
\begin{tabular}{l lll}
\toprule
\multicolumn{4}{l}{Dependent variable: Correctly identify news}\\
\hline
 & Fake & True &  All \\
 &(1) & (2) & (3) \\
\toprule
\textit{Panel A: Label 59\% of news items as fake}\\
\\
Treatment 2	 & -0.215***&	-0.681***&	-0.379***  \\
& (0.051)&	(0.038)	&(0.039))\\
\hline
\textit{Panel B: Label 75\% of news items as fake}\\
\\
Treatment 2	 & 0.169***&	-1.045***&	-0.260*** \\
& (0.050)&	(0.036)	&(0.039)\\
\hline
\textit{Panel C: Label 90\% of news items as fake}\\
\\
Treatment 2	 & 0.467***&	-1.292***&	-0.154*** \\
& (0.048)&	(0.033)&	(0.036)\\
%Observations & 8,111&4,428&12,539\\
\bottomrule
\end{tabular}
\begin{tablenotes}
    \footnotesize
      \item Notes: See text for construction of dependent variable.  All regressions include baseline controls and endline dummy. Treatment 1 indicator is not significant in any regression. Standard errors clustered by sample grid in parentheses.* p$<$0.1,  ** p$<$0.05,*** p$<$0.01.
\end{tablenotes}
\end{threeparttable}
%\vspace{-0.1in}
\end{table}

As a second robustness check, we carry out a falsification test to ascertain whether treatment 2 just increased the probability of labeling more news items as fake without making the treated individuals' judgements more accurate. In this test, we simulate a treatment group in which participants randomly identify news items as fake.
Note that the base rate of calling any news item fake was 59\% in treatment 2 compared to 54\% in the control group. We create three simulated dependent variables in which the responses of control and treatment 1 are kept intact, whereas the responses of treatment 2 individuals are replaced by the simulated responses with the rate of saying fake fixed at 59\%, 75\%, and 90\%, respectively. We standardize these dependent variables by the mean and standard deviation of the control group.

We find that if participants in treatment 2 had been randomly guessing 59\% of the news items as fake, the treatment effect on fake news would have been significantly negative, as respondents would have correctly identified fewer items as fake than the control group. If we raise the rate of saying fake in the treatment group to 75\%, then we observe a treatment effect on fake news comparable to the actual treatment effect. However, in this case, the discernment of true news is significantly lower than that in the control group, as the simulated treatment individuals are also labeling more true news as fake. Finally, if we raise the rate of saying false in the treatment group to 90\%, there is a large positive effect on fake news but a significantly larger negative effect on true news. Thus, the results in Table \ref{table:falsification1} support the interpretation that treated individuals are not randomly guessing about the accuracy of news items, in response to being primed about fake news or trying to act in conformance with the experimenters objective. 

Finally, we also analyze the treatment effect on placebo news items. These are synthetically constructed fake news items, which were never circulated on social media or any other media. There are two placebo news items shown to respondents in the first endline survey, as part of the list of news fake and true stories. Each placebo news item contained highly implausible news but also included an (invalid) Web link highlighting the source of news, which could only be checked for accuracy by visiting the link. If respondents were simply terming more news items as fake, then we would expect a positive treatment effect on placebo news, that is respondent identifies news as fake more often than the control group. While around 90\% of respondents (in both control and treatment groups) termed placebo news items as fake, we find a small negative effect (though not significant) on correctly identifying placebo news as fake news (see online Appendix Table \ref{table:effect_on_placebo}).
We do not find any effect on secondary outcomes such as engagement with the news or reasons for terming news as fake.
This suggests that experimenter demand effects are unlikely to explain treatment 2 effects.

Overall, these three robustness checks provide evidence against the explanation that experimenter demand effects are driving the treatment effect.

\subsection{Heterogeneous Effects}
\label{sec:het_eff}

\begin{figure}[t]
    \centering
    \captionsetup{font=footnotesize}
    \includegraphics[width=0.9\textwidth]{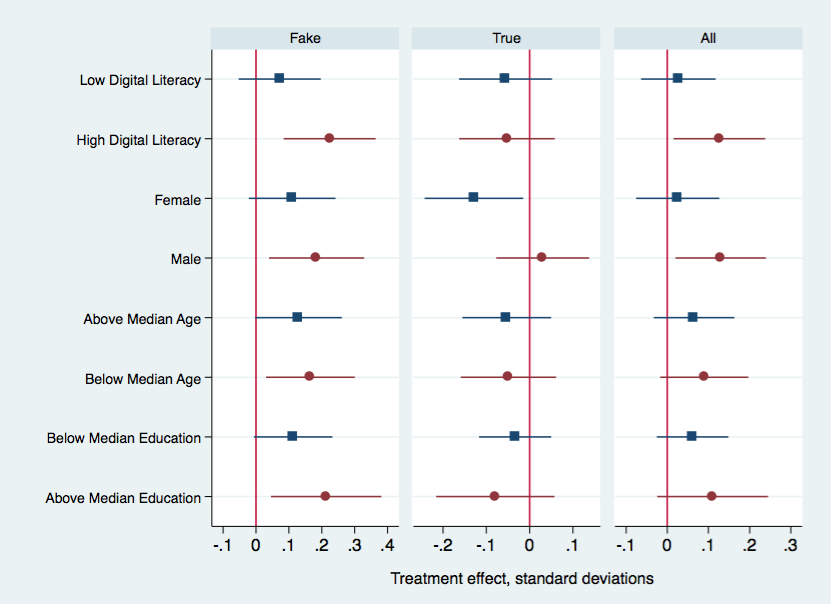}
    \caption{Treatment Effect by Digital Literacy and Demographic Characteristics}
    %\begin{figurenotes}
    \caption*{\footnotesize{Notes: The figure shows the effects of treatment 2 by digital literacy score, gender, age, and education. High and low digital literacy is defined as one standard deviation above or below mean digital literacy score. Median age is 27 years and median education level is Intermediate/Grade 12.  All regressions include standard controls. Error bars show 95\% confidence interval, with clustered standard errors.}}
    %\end{figurenotes}
    \label{fig:het}
   % \vspace{-0.1in}
\end{figure}

We now investigate the existence of heterogeneous treatment effects by four baseline characteristics: digital literacy, gender, age, and education. Figure \ref{fig:het} displays the heterogeneous effects, for each sub-group defined by these characteristics, separately for fake, true, and all news pooled together. We discuss these results below.

Our primary measure of digital literacy is the basic digital literacy score, which is an aggregation of the four questions used to capture Web-use, technical skills as well as the ability to understand and evaluate information online (as described in Section \ref{sec:context}).  We standardize the basic digital literacy score by the mean and standard deviation of the control group. Figure \ref{fig:het} shows the treatment effect estimates for individuals with high and low digital literacy (defined using one standard deviation above and below the mean respectively). The effect on the sub-group of low digital literacy individuals is positive at 0.07 standard deviations but not statistically significant. On the other hand, the effect on the sub-group of high digital literacy individuals is almost twice as large as the average treatment effect. There is no effect on true news for both sub-groups. Overall, the treatment effect is positive and significant for high digital literacy individuals when we pool all news together, which suggests improved discernment between true news and fake news.

In online Appendix Table \ref{table_whatsapp_fb}, we further assess the heterogeneity of treatment effect by two additional measures of digital literacy, namely the WhatsApp and Facebook scores which capture application-specific knowledge and usage. These indices are effective at capturing the variations across individuals with high basic digital literacy scores. We find that treatment 2 has a positive and significant effect on individuals with higher WhatsApp scores. The interactions of treatment 2 (and even treatment 1) indicators with the score are positive and statistically significant. We do not find any significant differential effect by Facebook score. Facebook users only comprise half of our sample and they are likely to be more educated and more experienced social media users, which can explain why the effects are no different for those with higher Facebook scores. 

Separating the treatment effect by gender shows that males benefit more from the treatment, both in terms of recognizing fake news and distinguishing between fake news and true news.
The effect on fake news for females is smaller relative to males and statistically insignificant, while the effect on recognizing true news is negative, suggesting greater skepticism in females regarding the accuracy of news (see Figure \ref{fig:het} and online Appendix Table \ref{table:hetero_by_demo}). 

One reason for the lower take-up can be that females also had lower levels of digital literacy. \footnote{According to an ITU study conducted between 2016-18, females are 11.6\% less to possess basic ICT skills (e.g., copying or moving a file or folder, sending emails with attachments, and connecting and installing a new device) compared to males in 9 low-income and lower-middle income countries for which the data was available \cite{itu-gender}.} While females have the same years of formal education as males and 69\% of them report being able to read English on social media versus 57\% of the males, on other dimensions of digital literacy they lag behind males. For example, 75\% report being able to connect to WiFi and mobile data (versus 91\% of males), 74\% report being able to use online search (versus 83\% of males), and 82\% report being able to use social media without assistance (versus 88\% of the males). Thus females are likely to have less experience in navigating social media and assessing the credibility of information encountered online. 

To check whether low digital literacy could be mitigating the treatment effect for females, we check for heterogeneous effects by digital literacy for females. Online Appendix Table \ref{table:hetero_for_females_by_dl} shows that the treatment effect using the interaction between treatment and digital literacy score for the female sample only. We find positive interactions of both treatment indicators for fake news items, suggesting greater ability of high digital literacy females to recognize fake news. For true news, we observe that the second treatment had a negative and significant effect which is an average effect. There is no significant difference for females by digital literacy. Taken together, these additional results suggest that although higher digital literacy females are more likely to correctly identify fake news, on average there is also increased skepticism towards true news by females relative to the control group. We further confirm this finding in \ref{sec:discern}.

To study the average treatment effect by age and education, we divide our sample using the median age and education which is 27 years and Intermediate/Grade 12, respectively. We find suggestive evidence of larger effects on fake news for younger and more educated sub-groups, no significant effect on true news, and no significant improvement in the overall ability to distinguish true and fake news (see Figure \ref{fig:het} and online Appendix Table \ref{table:hetero_by_demo}). \footnote{In additional results (available upon request), we explore whether age or education might be driving differences observed by digital literacy, by including interactions between age and education with treatment indicators, along with digital literacy interactions. The results indicate that digital literacy is the most important moderator, while the other two moderators are not significant.}

Taken together, these results suggest that there is scope for designing customized educational interventions that are specifically targeted towards users at the lowest end in our sample in terms of digital literacy and females. For these sub-populations, the existing interventions may be complex and ineffective. As a result, designing interventions that require even lower cognitive effort and prior training in analytical reasoning can help them discern fake news better.

\subsection{Does Treatment 2 Improve Discernment?}
\label{sec:discern}
While there does not seem to be evidence for experimenter demand effects driving the treatment effect, respondents may still be terming news as fake without improvement in discernment and simply due to increased skepticism towards news on social media. To assess this, we separate the treatment effect due to the overall change in the probability of labeling a news item as fake from the effects due to the change in correctly identifying fake news items as fake. Specifically, we run a regression in which the dependent variable is equal to one if a respondent labels a news item as fake and the treatment indicator interacted with the news item being actually fake. 

In column 1 of Table \ref{table:sayfalse} where we pool all fake and true news items from both endlines, we find that on average the second treatment effect does raise the likelihood of labeling news as fake. However, the effect is not significant. The differential effect on fake news is 1.5 times larger than the average effect but also not precisely estimated. The total effect on fake news, is thus due to both an average increase in probability of terming a news item as fake and a differential effect. This is consistent with our previous findings in Table \ref{table:ate}, where we find suggestive evidence of negative effects on true news and Figure \ref{fig:het}, where we show that females are driving the negative and significant effect on true news. 

Therefore, in columns two and three of Table \ref{table:sayfalse}, we separate the treatment effect on males and females. We find that males are not just labeling more news as fake as a result of the treatment. The entire effect of the treatment on males is driven by improved discernment of fake news. However, for females the opposite is true.  There is no additional effect on correctly identifying a fake news item as fake. This lends support to the interpretation that male respondents became more discerning about fake news as a result of receiving the second treatment. Females do not experience any improved discernment and are just as likely to label true news items as fake. Moreover, we do not find any differential effect for females exhibiting social desirability trait (see online Appendix table \ref{table:effect_by_social_gender}). These results are suggestive of increased skepticism about news on social media amongst females after receiving the treatment, rather than purely experimenter demand effects.

\begin{table}[tbp]
\centering
\begin{threeparttable}
\footnotesize
\captionsetup{font=footnotesize}
\caption{Effect of Treatment 2 on Labeling News as Fake}
\label{table:sayfalse}
\begin{tabular}{l p{1.5cm} p{1.5cm}p{1.5cm}}
\toprule
\multicolumn{4}{l}{Dependent variable: Label news as fake}\\
\hline
 & All & Males &  Females \\
& (1) & (2) & (3)\\
\toprule

Treatment 2	 & 0.050&	-0.022&	0.129**  \\
& (0.037) & (0.051) & (0.052)\\

Treatment 2 x Fake 	 & 0.079&	0.184**&	-0.019  \\
& (0.073) & (0.095) & (0.089)\\

Treatment 2 Total Effect on Fake News & 0.129**&	0.162**	&0.110*\\
& (0.063) & (0.051) & (0.059)\\

Observations & 12,539 &6,098 & 6,441\\
\bottomrule
\end{tabular}
\begin{tablenotes}
    \footnotesize
      \item Notes: Sample is all news items shown in two endlines. Dependent variable is equal to one if respondent labels a news item as fake, standardized by mean and standard deviation of control group. All regressions include baseline controls and endline dummy. Total Effect is the sum of coefficient on treatment indicator and interaction between treatment and fake news. The indicator for treatment 1 and interactions with fake news are not significant.  Standard errors clustered by sample grid in parentheses.* p$<$0.1,  ** p$<$0.05,*** p$<$0.01.
\end{tablenotes}
\end{threeparttable}
\end{table}

\subsection{News Characteristics}
We now analyze the effect of our interventions by the characteristics of the news items shown in the endline surveys. We classify the news items along four dimensions related to the content and format of the news. To study the effect of content, we create an indicator equal to one if the news had political content. If individuals already hold prior beliefs about the topic, then we should expect the effect of our treatment to be mitigated for those individuals whose prior beliefs conform with the fake news. We also create a second indicator which is equal to one if the news is related to personal health, safety, and food. We group these topics as they convey information that can lead to immediate action by the individual. For these types of news items, we expect that if the individual takes the wrong action based upon the misinformation then they can incur substantial physical or psychological costs. We term these high-cost news items. We formalize these costs in our conceptual model discussed in Section \ref{sec:model}. 
To study the effect of news format, we create indicators for whether the news was conveyed in Urdu (as opposed to English), and whether the news contained only text (relative to a combination of text and image).
\begin{table}[t]
\centering
\begin{threeparttable}
\footnotesize
\captionsetup{font=footnotesize}
\caption{Treatment Effect by News Characteristics}
 \label{table:news_char}
\begin{tabular}{l llll}
\toprule
\multicolumn{5}{l}{Dependent Variable: Correctly identify news}\\
\hline

Characteristic:& Political&High Cost&Urdu & Text\\
&(1)&(2)&(3)&(4)\\
\toprule
%Treatment 1&-0.016&	-0.034	&-0.024	&-0.004\\
%&(0.050)&	(0.042)&	(0.053)	&(0.042)\\

%Treatment 1 x Characteristic& -0.014&	0.034&	0.013&	-0.054\\
%&(0.042)&	(0.054)	&(0.048)&	(0.066)\\

Treatment 2&0.101**	&0.005	&0.094*	&0.081**\\
&	(0.046)	&(0.039)	&(0.050)	&(0.041)\\
Treatment 2 x Characteristic& -0.082**&	0.175***&	-0.059&	-0.014\\
&	(0.040)&	(0.057)&	(0.042)	&(0.062)\\
\bottomrule
\end{tabular}
\begin{tablenotes}
      \footnotesize
      \item Notes: Sample is all news items shown in the two endlines. News characteristics are political content, high cost content, Urdu format, text only format. See text for details of detailed definition of characteristics. All regressions include standard controls. The indicator for treatment 1 and interactions with news characteristic are not significant. Standard errors clustered by sample grid in parentheses.  * p$<$0.1,  ** p$<$0.05,*** p$<$0.01.
     
\end{tablenotes}
\end{threeparttable}
\end{table}

Table \ref{table:news_char} shows that the treatment effect is indeed smaller for political news, suggesting that prior beliefs do have an offsetting effect. The treatment effect is larger for news items which we term as the high-cost news items. Furthermore, we also observe that individuals in the control group are more likely to fall for such types of fake news as compared to all other types of fake news shown in the endlines.\footnote{On average, 42\% of the individuals in the control group were not able to identify the high-cost fake news items, while 27\% were not able to identify the remaining fake news items.} This is a novel insight as it shows that people are more likely to fall for misinformation on matters that touch upon their everyday lives and behavior, but there is also greater scope for teaching people how to correctly identify such misinformation. Table \ref{table:news_char} also shows a smaller treatment effect for news that was presented in Urdu and news that was presented in text-only format. The latter finding suggests that non-textual elements of news can help individuals discern the accuracy of the news.

\subsection{COVID-19 Misinformation Follow-up}
To assess the longevity of our treatment effect and effects on behaviors outside the experimental setting, we conducted a follow-up phone survey focusing on assessing beliefs and behaviors related to the COVID-19. The phone survey was conducted in September 2020, fifteen months after the interventions. We were able to reach 89.3\% of the respondents in our original RCT sample. Online Appendix Table \ref{table:survey_response_rates} shows that attrition was random and unrelated to the treatment status of the individual. In this survey we assessed the accuracy of beliefs about news related to COVID-19. We asked people about seven popular fake news items and three true news items related to COVID-19. For each news item, individuals were first asked whether they had heard the news before. Then they were asked to state to what extent they agreed with the statement. They had the option to agree, disagree, or say they were unsure. In addition, respondents were also asked questions about information-seeking and preventative behaviors during the pandemic.

The survey was conducted in September 2020, when the infection rates had declined considerably in Pakistan from the peak daily rate of around 7000 cases to 500 daily cases \cite{covid19-pakistan}. We found that all three true news items (related to hand-washing, wearing a mask, and physical distancing) were universally known and everyone surveyed had heard them before. Moreover, 98\% of the individuals had accurate beliefs about the true news items. For the seven fake news items, we found that 81\% recalled having heard the news before, and 79\% had accurate beliefs, that is they disagreed with the fake news. Therefore, it appears that even though misinformation related to COVID-19 had been circulating widely, by the time we conducted our phone survey, a large fraction of individuals had formed accurate beliefs about such news, potentially leaving little room for a treatment effect to be observed if one exists.\footnote{Several large-scale COVID-19 awareness campaigns were launched by the Government of Pakistan in collaboration with various organizations using print, electronic, radio and social media across the country \cite{ncoc,unicef_zong}. From March 19 and May 7 in 2020, 1028.5 million COVID-19 awareness messages were sent in Urdu, English, and regional languages to mobile users across Pakistan. Among these included using Ring Back Tones (RBTs) to spread awareness messages, which were activated on 131.7\,million subscribers’ mobile devices \cite{pta_campaign}. According to a nationally representative survey conducted by Gallup Pakistan, ringtone messages had a positive impact on knowledge, perception, and behaviors related to COVID-19 in Pakistan \cite{caller_tunes}}
\begin{table}[t]
\centering
\begin{threeparttable}
\footnotesize
\captionsetup{font=footnotesize}
\caption{Treatment Effect in Covid Misinformation Follow-up}
 \label{table:covid_followup}
\begin{tabular}{l lll ll}

\toprule
\multicolumn{6}{l}{Dependent Variable: Accuracy of News}\\
\hline
& \multicolumn{2}{l}{Fake}& \multicolumn{1}{l}{True}&\multicolumn{2}{l}{All}\\
& (1)&(2)& (3)&(4)&(5)\\
%\cline{2-6}
\toprule
\\
Treatment 1&-0.041&0.020&	0.024&	-0.022&0.037\\
&	(0.049)&(0.056)	&(0.028)& 	(0.036)&(0.053)\\
Treatment 1 x Recall &&-0.078 & &&-0.069\\
&&(0.075)&&&(0.061)\\

Treatment 2&0.010&0.089*&	0.012&	0.010&0.099**\\
&	(0.052)&(0.050)	&(0.029)&	(0.038)&(0.046)\\
Treatment 2 x Recall &&-0.092 &&&-0.099* \\
&&(0.069)&&&(0.056)\\
Recall& & -0.439***&&&-0.245***\\
&&(0.053)&&&(0.045)\\
Observations&\multicolumn{2}{l}{4,347}&\multicolumn{1}{l}{1,863}&\multicolumn{2}{l}{6,210}\\
\bottomrule
\end{tabular}
\begin{tablenotes}
      \footnotesize
      \item Notes: Sample is news items in the COVID-19 phone survey. See text for details about the construction of the dependent variable. Recall is an indicator equal to 1 if recalled seeing or hearing news before. Columns 2 and 5 estimated with interactions of treatment variables with recall. All respondents report recalling true news, and interactions are omitted. All regressions include standard controls. Standard errors clustered by sample grid in parentheses. * p$<$0.1,  ** p$<$0.05,*** p$<$0.01.
\end{tablenotes}
\end{threeparttable}
%\vspace{-0.1in}
\end{table}

In Table \ref{table:covid_followup} we estimate the effect of our treatments on the accuracy of beliefs about COVID-19 related news. Accuracy is equal to 1 if the individual correctly disagrees with a fake news item, 0.5 if they say they are not sure and 0 otherwise. Similarly, for true news, accuracy is equal to 1 if the individual correctly agrees with a true news item, 0.5 if they say they are not sure and 0 otherwise. As before, we standardize the dependent variable using the mean and standard deviation of the control group. Columns 1, 3, and 5 show the treatment effect on fake, true, and all news items respectively using our main specification with baseline controls. We do not find any significant effect of the treatment on the accuracy of beliefs. Given that our follow-up was conducted after the end of the first COVID-19 wave, was our sample (in control and treatment) generally aware of the accuracy of news related to COVID-19, and therefore, there is no effect of our treatment? For true news, all news items were recalled by everyone whereas 81\% of individuals had heard about the fake news items, signifying the general prevalence of public health information related to the pandemic. A related explanation for this prevalence is that the high personal cost (e.g., harmful health effects) associated with not having accurate beliefs about COVID-19 led individuals in both treatment and control groups to exert greater effort to form accurate beliefs when they heard such news. Another possibility is that the treatment effect simply faded over time.
%A related explanation is that due to the high personal cost associated with not having accurate beliefs about COVID-19 (e.g., harmful health effects), which led individuals in both treatment and control groups to exert greater effort to form accurate beliefs when they heard such news. Another possibility is that the treatment effect simply faded over time.

To help differentiate between the above two explanations, and assess whether general awareness is playing a role for fake news, we estimate the treatment effect separately for two sub-groups: those who recall and do not recall hearing the news before. Since recall is a post-treatment variable, we first examine whether recall rates of COVID-19 news items are affected by the treatment in online Appendix Table \ref{table:te_recall}. We do not find any significant effect of either treatment on the probability of an individual recalling a COVID-19 related fake news. This indicates that awareness of COVID-19 news was generally prevalent in our sample and the interventions did not have a differential effect on recall, thereby making it balanced across the groups. This gives us confidence that we can study the treatment effects on COVID-19 fake news by recall. The results are reported in column 2 and 5 of Table \ref{table:covid_followup} for fake news and all news respectively.

The coefficient on the recall variable is negative, which indicates that in the control group individuals who had previously heard a fake news item had less accurate beliefs relative to those who do not recall seeing a news item before.  Among the individuals who recalled having heard the fake news before, it is likely that some of them had formed correct beliefs, but there were others who had formed inaccurate beliefs which were not updated over time. These might be individuals who have strong prior beliefs which are consistent with the fake news and therefore, they are less likely to form correct beliefs (see Section \ref{sec:model}). There was no significant difference across treatment and control groups within this sub-group. This suggests that effects due to general awareness about fake news had permeated both groups, but there were some individuals who did not update their beliefs possibly due to strongly held prior beliefs. 

Furthermore, we also find that among those who did not recall hearing the fake news item before, treatment 2 improved their accuracy of beliefs by 0.09 standard deviations relative to the control group (significant at the 10 percent level). This suggests that the effect of treatment 2 on those who did not recall seeing a fake news item before, or in the absence of general awareness, though smaller, likely persisted over time.

\begin{figure}[t]
    \centering
    \captionsetup{font=footnotesize}
    \includegraphics[width=0.8\textwidth]{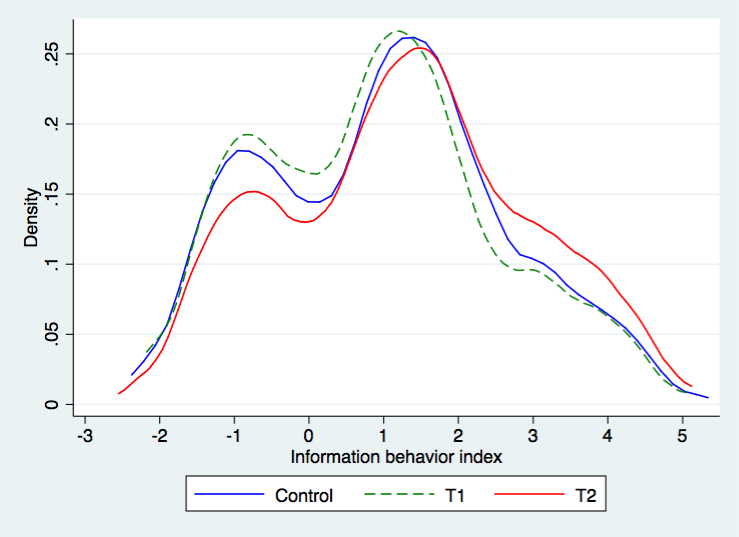}
    \caption{COVID-19 Information-seeking and Sharing Behaviors}
    %\begin{figurenotes}
    \caption*{\footnotesize{Notes: The figure shows the kernel density plots of a composite index of information-seeking and sharing behaviors of control, treatment 1 and 2 as measured in the COVID-19 phone survey. Index consists of five variables measuring frequency of use of formal news sources for COVID-19 information, trust in formal news sources for COVID-19 information, frequency of use of social media sources for COVID-19 information, trust in social media sources for COVID19 information, and sharing COVID-19 information on social media. }}
    %\end{figurenotes}
    \label{fig:covid-infoseek}
   % \vspace{-0.1in}
\end{figure}

\begin{figure}[t]
    \centering
    \captionsetup{font=footnotesize}
    \includegraphics[width=0.8\textwidth]{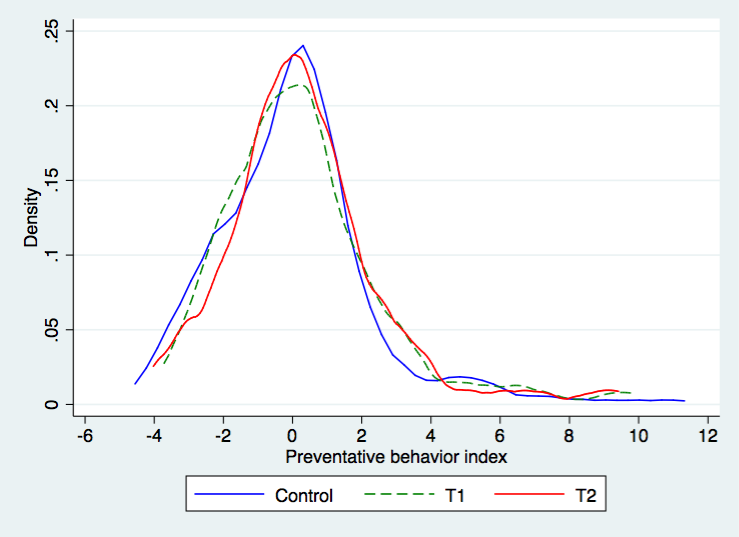}
    \caption{COVID-19 Preventative Behaviors}
    %\begin{figurenotes}
    \caption*{\footnotesize{Notes: The figure shows the kernel density plots of a composite index of preventative behaviors of control, treatment 1 and 2 as measured in the COVID-19 phone survey. Preventative behaviors index consists of four variables measuring social distancing, attitude towards sending children to school, mask ownership, and perceived market prices of masks.}}
    %\end{figurenotes}
    \label{fig:covid-sop}
   % \vspace{-0.1in}
\end{figure}

Next, we assess whether the treatments had any effect on  self-reported information-seeking and sharing behaviors of respondents related to COVID-19 (see Figure  \ref{fig:covid-infoseek}). We construct an \emph{information behavior index} using a list of five survey questions that measure the frequency of use of formal news sources (World Health Organization and government sources) and social media news sources (WhatsApp, Facebook, Twitter, and YouTube) for obtaining information about COVID-19, trust in these sources for providing reliable COVID-19 related information, and frequency of sharing COVID-19 information on social media. We code these variables such that the index takes increasing values the higher is the frequency of use of formal sources, trust in formal sources increases, the lower is the frequency of use of social media sources and trust in social media, and lower is the frequency of sharing COVID-19 related information on social media. Each of these variables is standardized using the mean and standard deviation of the control group. We then take the sum of the standardized variables to create the index. The questions corresponding to each of these measures and their coding is explained in online \ref{sec:app_indices}. 

We find that respondents in the treatment 2 group use formal sources more, have greater trust in formal news sources, and are less likely to share COVID-19 information on social media relative to the control group. There is no difference in frequency of use and trust in social media sources. When combined as an index, we find that the difference in the means of the index for treatment 2 and control is positive (0.26 standard deviations) and significant at the 10 percent level. As shown in \ref{fig:covid-infoseek}, we find that the distribution of the information behavior index is shifted rightward for the treatment 2 group relative to the control group. This reflects greater usage and trust in formal sources to obtain COVID-19 related information and a lower frequency of sharing COVID-19 related information on social media.

Finally, we construct a \emph{preventative behavior index} to measure respondents' willingness to observe (self-reported) COVID-19 preventative behaviors. We use a list of four survey questions that measure behaviors related to mask ownership and social distancing. Specifically, we use questions asking about ownership of different types of masks, the market price of different types of masks, interacting with people outside the house (in close proximity within 2 arm's length), willingness to send children to school if they reopen. Each of these variables is coded such that it takes increasing value the more preventative is the behavior. For price, we take higher prices as corresponding to higher willingness to pay and thus greater inclination towards preventative behaviors. The questions corresponding to each of these measures and their coding are reported in online \ref{sec:app_indices}. We standardize individual variables by the mean and the standard deviation of the control group and aggregate the standardized variables to create a preventative behaviors index. 

We find that individuals in the second treatment group are likely to own masks, observe social distancing, and have a lower willingness to send children to school. However, their reported market price of masks is lower than that of the control group. When combined as an index, we find that the distribution of preventative behaviors is rightward shifted and the mean of the index is higher for the treatment 2 group relative to the control group although the difference in means is not statistically significant (Figure \ref{fig:covid-sop}). Overall, we find suggestive evidence that individuals in the treatment 2 group are more likely to engage in preventative behaviors. Therefore, we conclude that while it is challenging to make a definitive conclusion about the effectiveness of treatment 2 on the perceived accuracy of COVID-19 related fake news, there is evidence to suggest that treatment 2 had an impact on information-seeking and preventative behaviors related to COVID-19.

\section{Mechanisms}
\label{sec:model}
In this section, we propose a conceptual framework to explain how educational interventions change an individual's accuracy of beliefs regarding news seen on social media. Our conceptual framework combines motivated reasoning and analytical reasoning, which are commonly used to explain how individuals assess information and form beliefs. We provide evidence consistent with the model's predictions using the data from our experiment.

\subsection{Conceptual Framework}
We develop a conceptual framework to model factors that determine how individuals form beliefs about the accuracy of news encountered on social media. Two broadly accepted influences on the beliefs individuals develop about the accuracy of news include motivated reasoning and analytical reasoning.
The effects of motivated reasoning on different forms of judgement have been widely documented and supported by prior studies \cite{why_reason,motivated_reasoning,hugh_right,kahan_2007,partisanship}.
For example, individuals are more likely to uncritically accept claims that support their political ideology compared to evidence that is inconsistent with their prior belief \cite{strictland}. Thus, motivated reasoning can at least partly explain why individuals may be susceptible to fake news, by accepting stories that conform to their prior beliefs.
Another growing body of work in the tradition of dual-process theories of human cognition argue that deliberative and analytical reasoning can lead to sound judgement \cite{kohl,evan_analytical,stan_analytical,crt,3_stage,rand_lazy}. Thus, applying analytical reasoning can increase the accuracy of beliefs about fake news and true news.

Our framework \emph{combines} these views by modelling individuals as motivated reasoners, i.e., they like news conforming to their prior beliefs but are also able to apply analytical reasoning to form beliefs about the news.
As a result, we show that an individual can correctly identify fake news even if it perfectly conforms to her prior belief if the cost of not finding the `truth' is high or the cost of analytical reasoning required to discern the truth is low. For example, an individual will expend more time and effort verifying a fake news item related to COVID-19 (which may have serious health implications for the individual) compared to a fabricated story about the chances of rain during the week.\footnote{Another example relates to \emph{social distancing} behaviors during the COVID-19 pandemic. An individual might engage in less social distancing because their cost of distancing is greater (e.g., the individual would lose more income) or because the benefits of distancing are smaller (e.g., the individual is at a lower risk of infection) \cite{polarization_health}.}
Similarly, individuals having higher levels of digital literacy, who can engage in fact-checking, will pay a lower cost for applying analytical reasoning as they are likely to arrive at the truth with less time and effort. In both cases, individuals will form correct beliefs about the fake news even if it conforms to their prior beliefs.

A key insight of our model is that by raising digital literacy or lowering the cost of analytical reasoning, an individual's accuracy of beliefs about news can be improved. Therefore, in our model educational treatments work by lowering the cost of analytical reasoning by providing individuals with simple decision rules to evaluate the credibility of news without expending significant time and effort.

\subsubsection{Individual Utility}
We build upon \cite{market_news} to develop a simple model to explain how individuals determine beliefs about a news story.
Individuals receive a news item, $N$, about a topic of interest (e.g., state of the economy, COVID-19), which can be true or fake, and individuals form beliefs about the accuracy of the news item. We assume that $N$ is drawn from the distribution $\mathcal{N}(0,v)$. For simplicity, we assume that both true and fake news items are drawn from the same distribution. Individuals will have accurate beliefs about $N$ whenever they can correctly identify a fake news item as fake and a true news item as true. We also assume that individuals' prior beliefs, $B$, about the underlying topic follow the distribution $\mathcal{N}(b,v)$, where $b \neq 0$ captures the expected bias in their prior beliefs.
The individual gets a baseline utility $\bar u$ from reading the news item, $N$. Individuals like to read news that confirms their beliefs, therefore, they incur a disutility when $N \neq B$. The disutility increases with the square of the distance between the news and the individual’s beliefs with an elasticity of $\phi$. Thus, the direct cost of reading news inconsistent with one's beliefs is $\phi(N-B)^2$.

Individuals can exert analytical reasoning to determine the accuracy of the presented news. Analytical reasoning depends on cognitive effort and productivity, which captures how good the individual is at reasoning given the effort exerted. Individuals with high digital literacy are expected to have higher productivity and thus achieve a high level of reasoning with a low amount of effort. In order to identify fake news correctly, individual needs to exert a high level of analytical reasoning. Conversely, if the individual exerts a low level of analytical reasoning, they form inaccurate beliefs about the news. The level of reasoning required to form correct beliefs will vary across fake and true news and within each category also vary with news characteristics. We denote cognitive effort  by $a^i_j$, where $i\in\{H,L\}$ and $j\in\{T,F\}$ corresponds to true news ($T$) and fake news ($F$). Effort is costly and the individual incurs a total cost $\xi (a^i_j)^2$, for exerting $a^i_j$ units of effort to determine the accuracy of the news. Note that the marginal cost of effort increases with $\xi$, while the productivity of effort declines with $\xi$.

We assume that individuals care about knowing the truth. So when they form correct beliefs by applying a high level of analytical reasoning, they receive the benefit of knowing the truth. Benefit has a fixed component and a variable component, which depends on the degree of inconsistency between the news and prior belief. Specifically, we assume that the benefit of correctly identifying news is $\delta (1+(N-B)^2)$ if the news is fake and $\frac{\delta} {(1+(N-B)^2)}$ if the news is true. These terms reflect that individuals derive the maximum benefit if their prior belief is confirmed after reasoning. Therefore, for fake news, the benefit is lowest if $N=B$ (i.e., when prior beliefs are consistent with fake news but the news was correctly identified as fake). For true news, the benefit is maximized if $N=B$ (i.e., when prior beliefs are consistent with the true news item that was identified as true). 

If the individual has incorrect beliefs, they incur the cost of not knowing the truth. This can be a psychological or physical cost that arises due to inaccurate beliefs such as harmful actions taken due to believing and acting upon misinformation. For fake news, the cost is modeled as  $\chi(1+(N-B)^2)$, which is increasing in the distance of news from prior belief. For true news, we assume that this cost is $\frac{\chi}{(1+(N-B)^2)}$, which is decreasing in the distance of news from prior beliefs.

With this setup, let $U^i_j$ be the utility of concluding that the news item is $i\in\{T, F\}$, when the underlying state is $j\in\{T, F\}$. 

The utility of concluding that a fake news item is fake is given by:
\begin{equation}
U_F^F = \bar{u} - \phi(N-B)^2 -\xi (a^H_F)^2 +\delta(1+(N-B)^2)
\end{equation}
And the utility of incorrectly concluding that a fake news item is true is given by:
\begin{equation}
U_F^T = \bar{u} - \phi(N-B)^2 -\xi (a^L_F)^2 - \chi(1+(N-B)^2)
\end{equation}
Note that in this case, the individual incurs the cost of not knowing the truth, which increases with the distance between and $N$ and $B$. There is no additional cost of overcoming bias because the reader believes the fake news in this case.

Similarly, the utility of concluding a true news item as true is given by:
\begin{equation}
U_T^T = \bar{u} - \phi(N-B)^2 -\xi (a^H_T)^2 + \frac{\delta}{1+(N-B)^2}
\end{equation}
Note that when the news is true and it aligns with the individual's prior belief, no additional cost is incurred. Moreover, the derived benefit is decreasing in the distance between the news and the prior belief. 

And finally, the utility of concluding a true news item as fake is given by:
\begin{equation}
U_T^F = \bar{u} - \phi(N-B)^2 -\xi (a^L_T)^2 - \frac{\chi}{1+(N-B)^2}
\end{equation}

\subsection{Discerning between Fake News and True News}
Let $p$ be the probability that a user will correctly identify a fake news item as fake. This is the probability that utility of concluding that a fake news item is fake exceeds the utility of concluding that a fake news item is true.
\begin{equation}
    p = P(U_F^F > U_F^T) = P\left(-\xi a_H^2 +\delta (1+(N-B)^2)<-\xi a_L^2 -\chi(1+(N-B)^2)\right)
\end{equation}
\begin{equation}\label{eq:discern}
    = P\left( N-B > \sqrt{\frac{\xi A}{\delta + \chi} - 1} \right)
\end{equation}
where $A=(a^H_F)^2-(a^L_H)^2$.

Equation \ref{eq:discern} shows that $p$, the probability of correctly identifying fake news as fake, is increasing in $A$ and decreasing in both $\delta$ and $\chi$. Online Appendix B shows that using the underlying distributions of $N$ and $B$, the probability of correctly identifying a fake news can be written as:
\begin{equation}
    p = 1 - \Phi\left(\sqrt{\frac{\frac{\xi A}{\delta+\chi} - 1}{2v}}\right)
\end{equation}

%%----- Propositions -----%%
Next, we discuss the predictions of the model regarding the probability of correctly identifying fake news, as news characteristics and analytical reasoning are varied. We also examine the consistency of the model predictions with the experimental data.\footnote{We do not analyze $P(U_T^T > U_T^F)$ but it can be solved in the same manner as $P(U_F^F > U_T^F)$.}
%The effects are captured by the following propositions that summarize the outcomes for different cases. 

\subsection{Effect of News Characteristics on $p$}
We examine how the probability of correctly identifying fake news changes with three different characteristics of the news which are captured in our model, (i) the expected distance between news and the bias of the reader (ii) cost of not knowing the truth, and (iii) benefit of knowing the truth. Our results are stated in Propositions 1 to 3. See online Appendix B for all proofs. 

\begin{proposition}
When the expected bias $b$ decreases, the probability of correctly identifying fake news increases.
\end{proposition}
%\begin{proof}
%\,\,\,See Appendix B for all proofs
%\end{proof}

\begin{proposition}
\,As the cost of not knowing the truth, $\chi$, increases, the probability of correctly identifying fake news increases.
\end{proposition}

\begin{proposition}
As the benefit of knowing the truth increases, $\delta$, the probability of correctly identifying fake news increases.
\end{proposition}

In Table \ref{table:news_char}, we examined how the treatment effect varies with news characteristics. Here we draw the link between features of the model and our findings. First, we found a relatively small effect for political news. We expect that if individuals have strongly held prior beliefs in terms of their political leaning, this is likely to shift the mean of the belief distribution farther away from the mean of the news distribution. In the model, this translates into an increase in the bias in prior belief, which decreases the probability of correctly discerning political fake news. Second, we also found that the treatment effect is larger for news items related to health, safety, and food. We expect that believing fake news on such topics can have immediate negative effects if the individual takes the wrong action. Thus, such news items are likely to have a higher $\chi$ and therefore we observe individuals are more likely to identify these types of fake news. 

\subsection{Effect of Educational Interventions on $p$}

\begin{proposition}
As $\xi$ decreases  the probability of correctly identifying fake news increases.
\end{proposition}

\begin{proposition}
As the difference between $(a^H_F)^2$ and $(a^L_F)^2$ increases, the probability of correctly identifying fake news decreases.
\end{proposition}

In our model, $\xi$ is inversely proportional to the productivity of effort applied for reasoning. Our educational intervention improves the productivity of effort and decreases the marginal cost of applying analytical reasoning by providing simple decision rules to evaluate the credibility of sources as well as indicators of problematic content.
%In other words, individuals should incur a lower cost of analytical reasoning and thus should be more likely to identify fake news.
Moreover, by raising productivity, we also expect that the amount of effort required to correctly identify fake news should also decrease. This implies that the difference between $(a^H_F)^2$ and $(a^L_F)^2$ decreases, thereby increasing $p$. This is the short-run effect of our interventions. In the long run, educational interventions can also make people less biased and cause them to change their prior beliefs.

%\subsection{Evidence on how treatment affects model parameters}
In order to understand how effort changes as a result of our intervention, we use response time as a proxy for how much effort is devoted to answering the questions related to the beliefs about news items. We expect that if the intervention decreased $\xi$  and $a^H_F$, then the response time for each of the news items should fall. Data on response time by news item are available from time audit sensor data collected by the survey application during the interaction with the respondents. We are able to match all of the respondents in the first endline and 714 of the 744 respondents in the second endline to time audit data files available from the survey application. To exclude outliers, we trim the top 1\% observations and exclude negative response time observations. For the remaining sample, we check whether the treatment had any significant effect on the time taken by an individual to respond to the questions asked about each news item.
\begin{table}[t]
\centering
\begin{threeparttable}
\footnotesize
%\hline
\captionsetup{font=footnotesize}
\caption{Effect of Treatment on Response Times}
\label{table:effect_responsetimes}
\begin{tabular}{l p{1.5cm} p{1.5cm} p{1.5cm}}
\toprule
\multicolumn{4}{l}{Dependent variable: Response Time }\\
\hline
 & Fake & True &  All \\
 & (1) & (2) & (3) \\
 \toprule
\multicolumn{4}{l}{\textit{Panel A: All responses}}\\
%\hline
% & Fake & True &  All \\
%\cline{2-4}
Treatment 1&	0.035&	-0.031&	0.011\\
	&(0.062)&	(0.031)&	(0.043)	\\
Treatment 2	&	-0.092&	-0.033&-0.071*\\
&	(0.058)&	(0.028)	&(0.042)	\\
	Observations&	7,817&	4,301& 12,118\\
\hline
\multicolumn{4}{l}{\textit{Panel B: Correct responses}}\\
%\hline
%& Fake & True &  All \\
%\cline{2-4}
Treatment 1&	-0.005&	-0.032&	-0.016\\
&	(0.056)	&(0.031)	&(0.039)\\
Treatment 2	&	-0.081&	-0.006&	-0.050\\
		&	(0.052)&	(0.029)&(0.038)\\
	Observations&	5,620	&3,121&8,741\\
\hline

\multicolumn{4}{l}{\textit{Panel C: Incorrect responses}}\\
%\hline
%& Fake & True &  All \\
%\cline{2-4}
Treatment 1&	0.101&	-0.028& 0.061\\
&	(0.095)&	(0.064)&(0.070)\\
Treatment 2&	-0.093	&-0.103*&-0.128**\\
&	(0.089)	&(0.054)&(0.064)\\
Observations &	2,197&	1,180&3,377\\
\bottomrule
\end{tabular}
\begin{tablenotes}
    \footnotesize
    \item Notes:  Dependent variable is standardized to the the mean and standard deviation of control group. All regressions include baseline controls and endline dummy. Standard errors clustered by sample grid in parentheses. * p$<$0.1,  ** p$<$0.05,*** p$<$0.01.
\end{tablenotes}
\end{threeparttable}
%\vspace{-0.1in}
\end{table}

In Table \ref{table:effect_responsetimes}, we investigate the effect of our treatment on response time for questions related to the news items shown at both the endline surveys. The mean response time for all news items in our sample is 93.7 seconds and the standard deviation is 83.6 seconds. In Panel A we report the effect for fake, true, and all news items separately. We find that treatment 2 reduced the response time by 0.073 standard deviations relative to the control group (statistically significant at the 10 percent level) and this effect is larger (although not significant a the 10 percent level) for fake news items. This suggests that treatment 2 works by increasing the productivity (or lowering the cost) of analytical reasoning by giving individuals simple decision rules to help them discern between fake news and true news. As the cost of analytical reasoning declines, individuals are more likely to have correct beliefs. For example, \citeN{rubin} argues that when presented with difficult or complex tasks, invoking simple decision rules can reduce response time as well as mistakes especially when such rules are correlated with making fewer mistakes, e.g., simple rules for spotting fake news.
However, when the notion of a mistake is clear cut (e.g., counting the letter F in a paragraph) then short response times can lead to greater mistakes.

Next, we separate the effect of treatment on response time separately for correct and incorrect responses. Panel B shows the effect of treatment on those who correctly identified a news item. Again we find suggestive evidence of a reduction in the time spent by the second treatment group (although not significant). It is important to note that individuals in the control group who are correctly identifying fake news are likely those with high digital literacy or whose prior beliefs are inconsistent with the fake news. Therefore, they are likely to spend less time answering and also more likely to correctly answer. Despite the composition heterogeneity, we find suggestive evidence that on average, individuals in the second treatment group are spending less time when they correctly identify fake news. 

In Panel C, we study the sample with incorrect responses. We find that treatment 2 had a negative effect on the time spent answering fake and true news related questions (significant for true news). Again there are likely to be composition differences in the type of individuals in the control and treatment groups who answer incorrectly. Individuals in the treatment group who answer incorrectly are likely to be those for whom the news conforms to their prior beliefs. Or this group may also comprise very low digital literacy individuals who experience smaller benefits of the treatment. In both cases, individuals in the second treatment group spend less time analyzing the news (relative to the control group) and form incorrect beliefs.

\section{Discussion}
\label{sec:discussion}
In this section, we discuss the differences between our two educational interventions, the scaling up and external validity of the second treatment.

\subsection{Features of Intervention}
Our results show that treatment 2 has a positive effect on improving individuals' ability to identify fake news. However, educating users about misinformation through general educational messages does not have a significant effect on improving their beliefs about the accuracy of fake news.
There are two features of treatment 2 that distinguish it from treatment 1:
(i) \emph{personalized feedback}, which depends on an individual's own performance in the baseline test, and
(ii) \emph{guidance} provided by the enumerators who point out different features of misinformation.
\begin{table}[t]
\centering
\begin{threeparttable}
\footnotesize
%\hline
\captionsetup{font=footnotesize}
\caption{Effect of Treatment by Baseline Score }
\label{table:effect_by_baselinescore}
\begin{tabular}{l p{1.5cm} p{1.5cm}p{1.5cm}}
\toprule
\multicolumn{4}{l}{Dependent variable: Correctly identify news}\\
\hline
 & Fake & True &  All \\
& (1) & (2) & (3) \\
\toprule
%Treatment 1 & -0.018 & -0.198&-0.019	\\
% & (0.062) & (0.041)& (0.047) \\
%Treatment 1 x Baseline score & -0.024 & -0.011 & -0.019\\
% & (0.039) & (0.031)& (0.029) \\

Treatment 2	 & 0.189* & -0.055& 0.0103 \\
 & (0.112) & (0.078)& (0.083) \\
Treatment 2 x High baseline score & -0.048 & 0.006& -0.029\\
 & (0.112) & (0.087)& (0.086) \\ 

Observations & 8,111&4,428& 12,539 \\
\bottomrule
\end{tabular}
\begin{tablenotes}
    \footnotesize
      \item Notes: High baseline score is an indicator for above median baseline score. Sample is all news items shown in first two endlines. All regressions include standard controls. Treatment 1 and its interaction is not significant. Standard errors clustered by sample grid in parentheses. * p$<$0.1,  ** p$<$0.05,*** p$<$0.01.
\end{tablenotes}
\end{threeparttable}
%\vspace{-0.1in}
\end{table}

Personalized feedback has been shown to be more effective in drawing attention than general informational messages \cite{personal1,personal7,personal5,personal6,personal3}. Moreover, corrective feedback can facilitate learning and boost long-term memory \cite{feedback_power} and its timing can play an important role in enhancing learning \cite{timing_matters}. In order to assess the importance of feedback, we estimate the effect of treatment 2 by baseline scores. Table \ref{table:effect_by_baselinescore} shows that as baseline score increases, the effect of treatment 2 decreases. This suggests that feedback was more valuable for those with lower baseline scores, while those with higher baseline scores experienced smaller gains.

Another feature of treatment 2 that can partly explain its effectiveness is the \emph{guidance} provided by enumerators while delivering feedback (e.g., pointing out different features of misinformation). Such guidance (or instruction) becomes more effective when the complexity of the presented material increases \cite{why_instruction,instruction_influence}.
As the median education level in our sample is Grade 12, some users may have found the video message complex or difficult to understand. In fact, our sub-group analysis shows a larger effect on more digitally literate users, which is consistent with the view of complementarity between prior knowledge and new learning.

\subsection{Scaling Up}
An important consideration for the delivery of educational interventions is the ease with which they can be delivered at scale. Treatment 2 can potentially be delivered at scale on social media platforms like Facebook and Twitter.
First, rather than showing fake and true news items to users, social media platforms can provide retrospective feedback to users based on their engagement with news items that \emph{later} turn out to be fake. Social media platforms regularly use fact-checking services on their platforms which can aid in this regard \cite{whatsapp-factorgs}. Second, a video can be shown to users on their news-feed in their own language. Finally, guidance can be provided by showing users the features found by fact-checking services (e.g., unverified source, altered images), which can be scalably converted into audio messages using text-to-speech and speech translation systems. 

\subsection{External Validity}
An important question about the external validity of our educational interventions is regarding the extent to which the effectiveness of treatment 2 can be generalized to a different setting. For example, individuals may learn how to identify fake news only when they are presented exactly the way we taught them (i.e., our interventions are ``teaching to the test"). To mitigate such concerns, we include news on a range of topics and in a variety of formats that are different than the three baseline fake news items shown to participants. We also conduct two endlines post-treatment and a phone survey fifteen months after the intervention. Our results suggest that the effect of treatment 2 was present across a range of news items presented in different formats and on behaviors likely to be affected by belief in COVID-19 related misinformation. The effect seems to have persisted even after a significant time had elapsed after the treatment was delivered.

However, we acknowledge that our randomized evaluation compares the difference between treatment and control group in a specific geographical area, population, and context, and our unit of observation is not large enough to pick up any general equilibrium effects. Therefore, we do not know the effectiveness of our interventions if implemented in a different context and at a larger scale \cite{meager}. Moreover, we do not observe real behaviors on social media, as our focus in this paper was low digital literacy users who are not well represented on platforms such as Twitter where it is easier to observe sharing and engagement with news. Therefore, going forward, it would be useful to replicate such a study at a larger scale, and with outcome measures that are reflective of online and offline behavior in various settings where people receive and process information (e.g., politics, health, economics, etc). 

%While replication is important, it may never be feasible to rigorously test the extent to which our research results hold in all possible situations. However, our experiment does shed light on both the effectiveness of user-focused educational interventions to combat misinformation and the features of such programs that are likely to work.

Finally, in this work, we do not capture the welfare effects of our interventions such as implications of being more skeptical about news in general, impact on personal well-being and mental health, and the effect on social or political outcomes such as trust, coordination, and polarization. Measuring the welfare effects of such programs is an important area for future work.

\section{Conclusion}
\label{sec:concl}
As digital spaces become an important avenue for news consumption, it is important to understand how we can minimize the impact of misinformation on individuals and society. This is an even more challenging problem in contexts where individuals with limited digital literacy skills are joining online platforms in increasing numbers. Educational interventions focusing on increasing the ability of individuals to identify fake news offer a promising approach to combat misinformation, though evidence on their effectiveness is limited. 

In this paper, we designed and evaluated the effectiveness of two educational interventions to counter the spread of misinformation in a randomized control setting. 
%\textcolor{red}{We carry out our interventions with a sample of social media users drawn from low and middle income urban areas in a developing country, where we observe high rates of social media use but substantial variation in digital literacy. }
Our findings show that educating individuals about how to recognize misinformation through a general video-based message does not have any significant effect. However, showing the educational video and providing personalized feedback about individuals' past engagement with misinformation, and guidance on how to recognize fake news, significantly increases the probability of identifying fake news by 0.14 standard deviations and both fake and true news by 0.076 standard deviations. 
%Treated individuals are less likely to report extreme reactions, more likely to mention a source problem and less likely to mention prior beliefs as the reason for their views about the accuracy of news.  
We find that while there is no effect on identifying true news for males, females are likely to become skeptical even about true news. We conduct several robustness checks to show that experimenter demand effects are not driving our results.

Our findings have important policy implications for combating misinformation in the digital space among low digital literacy populations. We show that while educational programs hold promise, the design and delivery of such programs matter. While general educational messages delivered once do not produce any significant effect, personalized educational messages with feedback and guidance are successful in helping individuals learn how to identify misinformation. However, we do find evidence of lower take-up among females and less digitally literate individuals, which highlights the importance of designing programs that appeal to these segments and help them learn effectively.

% modest? too much -- real behavior?
Despite these contributions, we acknowledge that our experiment is carried out in a particular setting, and we do not observe actual online and offline behavior. 
%It will also be important to implement such programs at a regional or national scale to understand whether the results generalize and if there are any general equilibrium effects that can increase or decrease the total impact of educational programs in reducing misinformation. 
We are also cognizant that there are other dimensions on which educational programs can be customized for effective learning, for example, perceived trustworthiness of who delivers the message, interactivity, and repetition. These remain important areas for future work. Finally, while educational efforts to improve digital literacy are certainly not a panacea to the misinformation problem, they can prove to be a vital way to reduce vulnerability to fake news and thereby improve the informational well-being of individuals and society at large.

\section*{Administrative Information}
\noindent
\emph{Ethics and IRB Review.}
Our study was approved by the IRB at Lahore University of Management Sciences with protocol number LUMS IRB/03252019).

\medskip
\noindent
\emph{Funding.}
We gratefully acknowledge support by the Facebook Integrity Foundational Research Award 2019 from Facebook Research in the form of an unrestricted gift and the Lahore University of Management Sciences (LUMS) Faculty Initiative Fund grant (FIF 597-2020-ECO).

\medskip
\noindent
\emph{Declaration of Interest.}
None.

\bibliographystyle{chicago}
\begin{footnotesize}
\singlespacing{
\bibliography{paper}}
\end{footnotesize}

\clearpage
\appendix
\section{Sampling, Survey Procedures and Randomization}
We implement our survey and randomized control trial (RCT) in the city of Lahore in Pakistan which has a population of approximately 11 million.
In order to draw our sample of low and middle income households, we rely on population density data. The pattern of urban development in Lahore is horizontal, and dwelling size is inversely correlated with income levels, therefore, we use population density as a proxy for income levels. Specifically, we use the \cite{asiapop} satellite data which provides population counts at a spatial resolution of 100m by 100m, to identify the low and middle income areas of Lahore from where we draw our sample. The selected the areas of Lahore account for 35\% of the city's total population and cover seven of the fourteen national assembly constituencies (see online Appendix Figure \ref{fig:samp1}). These areas cover the older parts of the city where the median population density is 109 persons per 100m by 100m grid.\footnote{Urban development in Lahore has led to new housing schemes being developed outside the older parts of the city, and the richer and affluent households have moved into these new housing schemes.}  As a comparison, the median population density in the areas of Lahore not covered by our sample is 28 persons per 100m by 100m grid. 

We draw a random sample of 200 grids from the selected areas and survey five households per grid, covering a total of 1000 households in the baseline survey. In order to initiate the data collection within the chosen grids, we randomly dropped a point (x and y coordinate) within the grid. The enumerators arrived at the point and used the left hand rule to survey, within each grid, five households  where at least one social media user is present. The definition of social media user for our study is that the respondent must be at least 18 years of age and use either Facebook or WhatsApp. 

The RCT sample consists of 150 grids drawn randomly from the baseline grid sample. We randomly assigned 50 grids each to control, treatment 1, and treatment 2 status. In order to ensure gender balance, we also pre-specified the number of male and female participants to be drawn from each grid for the RCT (either 2 or 3 out of a total of 5 individuals per grid point). At the start of the baseline survey, the enumerators were informed (by the survey application), whether a male or female respondent was to be selected for the experiment. If the required gender was not present, the household was replaced. This gives us a sample of 250 social media users per treatment arm, and 750 social media users altogether drawn from 750 unique households, who become part of our experiment and are visited again at endline. Online Appendix Figure \ref{fig:samp2} shows the distribution of the grid points coded by treatment status and Figure \ref{fig:randomization} shows the randomization scheme.

\subsection{Baseline Survey}
The baseline survey is answered by all social media users present in the household at the time of the survey. The survey consists of the following sections:
\begin{enumerate}
    \item Household and demographic characteristics including household size, monthly expenditures, number of social media users, age, gender, and highest education level. Age is measured as a continuous variable, while education is categorized by highest level completed (0=no education, 1=primary, 2= middle, 3= secondary/grade 8, 4=matric/grade 10, 5=intermediate/grade 12, 6=bachelors, 7=masters, 8=above masters).
    \item Social media use including phone ownership and sharing, social media applications used, number of hours spent per day using Facebook and WhatsApp, size and composition of the social media network, top three news sources, and sharing of news on social media. We also ask if they have heard of the term ``fake news" and if they ever believed a news that later turned out to be false.
    \item Digital literacy measured by self-reported ability to turn on mobile data, WiFi, use Google search, use social media without assistance, read text on social media, and read English and Urdu on social media. We also ask about knowledge and use of common features of WhatsApp and Facebook. 
\end{enumerate}

\subsection{Questions on News items}
After the end of the baseline survey, we show each social media user a series of news items. We compiled a list of seven recent popular true and fake news stories, covering current events and general interest topics. Out of these three are true news stories and three are fake news stories. We also created one placebo news item which is used to measure and control for false recall. The news stories are shown in the form of screenshots of messages, posts, or tweets, similar to how users would typically receive news on social media. Some of the news stories are in English only, while others are completely or partially in Urdu. The respondents are first asked to carefully view a printed version of the screenshot and then asked to listen to an audio recording of the news in Urdu. After having viewed and heard the news, we ask the following questions:
\begin{itemize}
\item Have you seen this news before? (Yes/No)
\item If yes then did you think this news was true the first time you saw it? (Yes/No)
\item Right now do you think that this statement is true? (Yes/No)
\item When you saw this news how did you feel? (Very positive, positive, neutral, negative, very Negative).
\item When you saw this news what did you want to do? (Nothing, ask friends and family about it, search online, other)
\item Why do you consider the news to be true or fake? (unprompted and select all relevant options from: aligns with my political beliefs, aligns with my religious beliefs, it has the right source (link/report/video/article), it uses professional language, the picture is real, it is not biased, other)
\item Will you share the news on social media? (Yes/No)
\end{itemize}

\subsection{COVID-19 Follow-up Survey}
The follow-up survey was carried out in September 2020 over the phone. Participants were paid Rs. 100 as mobile credit for completing the survey. The survey includes questions about sources of information about news related to COVID,  beliefs about a set of popular COVID related news statements, and questions on preventative behaviors. In order to measure their beliefs about COVID related news, we create a set of ten news statements (seven fake and three true). The statements are read out by the enumerator on the phone and the participant answers two questions about each statement, 
\begin{itemize}
\item Have you seen or heard this statement before? (Yes/No)
\item To what an extent do you agree with the statement? (Agree/Disagree/Not Sure)
\end{itemize}

Our dependent variable, which is belief about the accuracy of the news statement, takes the value 1 if the individual has correct beliefs about the news statement (that is agrees with true news and disagrees with fake news), 0 if they have incorrect beliefs, and 0.5 if they answer not sure. 

\newpage
\section{Proofs of Propositions}

%------- Proof of Proposition 1 -------%
\noindent
\textbf{Proof of Proposition 1}:\\
\begin{proof}
\, To show that $p$ increases with $A=(a^H_F)^2 - (a^L_F)^2$, we take the partial derivative of $p$ with respect to $a$, which gives the following expression:
\begin{equation}
    \frac{\partial p}{\partial A} = -\frac{\phi\left(D\right) \cdot \xi}{4v(\delta + \chi)D}
\end{equation}
where $D=\frac{ b + \sqrt{\frac{\xi A}{\delta + \chi} - 1} } {\sqrt{2v}}$, and $\phi(.)$ is the probability density function of the standard normal random variable. $\frac{\partial p}{\partial A} < 0$ because $\phi(.) > 0$ and we assume $D > 0$. Thus $p$ is decreasing in $A$.
\end{proof}

%------- Proof of Proposition 2 and 3 -------%
\medskip
\medskip
\noindent
\textbf{Proof of Propositions 2 and 3}:\\
\begin{proof}
\, To show that $p$ increases with $\chi$, we take the partial derivative of $p$ with respect to $\chi$, which gives the following expression\footnote{We obtain the same partial derivative with respect to $\delta$, the benefit of knowing the truth.}:
\begin{equation}
    \frac{\partial p}{\partial \chi} = \frac{\phi\left(D\right)}{4v(\delta + \chi)^2D}
\end{equation}
where $D=\frac{ b + \sqrt{\frac{\xi A}{\delta + \chi} - 1} } {\sqrt{2v}}$ and $\phi(.)$ is the probability density function of the standard normal random variable. $\frac{\partial p}{\partial \chi} > 0$ because $\phi(.) > 0$ and we assume $D > 0$. Thus $p$ is increasing in $\chi$.
\end{proof}

%------- Proof of Proposition 4 -------%
\medskip
\medskip
\noindent
\textbf{Proof of Proposition 4}:\\
\begin{proof}
\, To show that $p$ decreases with the expected bias $b$, we take the partial derivative of $p$ with respect to $b$, which gives the following expression:
\begin{equation}
    \frac{\partial p}{\partial b} = -\frac{\phi\left(D\right)}{\sqrt{2v}}
\end{equation}
where $D=\frac{ b + \sqrt{\frac{\xi A}{\delta + \chi} - 1} } {\sqrt{2v}}$, and $\phi(.)$ is the probability density function of the standard normal random variable. $\frac{\partial p}{\partial b} < 0$ because $\phi(.) > 0$ and we assume $D > 0$. Thus $p$ is decreasing in $b$.
\end{proof}

%------- Proof of Proposition 5 -------%
\medskip
\medskip
\noindent
\textbf{Proof of Proposition 5}:\\
\begin{proof}
\, To show that $p$ increases by decreasing $\xi$, we take the partial derivative of $p$ with respect to $\xi$, which gives the following expression:
\begin{equation}
    \frac{\partial p}{\partial \xi} = \frac{\phi\left(D\right)A}{4v(\delta + \chi)D}
\end{equation}
$\frac{\partial p}{\partial \xi} < 0$ because $\phi(.) > 0$ and $D=\frac{ b + \sqrt{\frac{\xi A}{\delta + \chi} - 1} } {\sqrt{2v}} > 0$. Thus $p$ increases when $\xi$ decreases.
\end{proof}

%%%%%%%%%%% BEGIN OLD Propositions %%%%%%%%%%%%%%%
\begin{comment}
\begin{proposition}
As $\delta$ increases the probability of correctly identifying fake news decreases.
\end{proposition}
\begin{proof}
\, To show that $p$ increases with $\delta$, we take the partial derivative of $p$ with respect to $\delta$, which gives the following expression:
\begin{equation}
    \frac{\partial p}{\partial \delta} = \phi\left(\frac{\frac{\chi(1+a)}{\delta}-b -c}{\sqrt{2v}}\right) \cdot \frac{-\chi(1+a)}{\delta^2 \sqrt{2v}}.
\end{equation}
$\frac{\partial p}{\partial a} < 0$ because $\phi(.) > 0$ and $\frac{-\chi(1+a)}{\delta^2 \sqrt{2v}} < 0$. Thus $p$ is decreasing in $\delta$.
\end{proof}
\end{comment}
%%%%%%%%%%% END OLD Propositions %%%%%%%%%%%%%%%

\clearpage
\section{Appendix Figures and Tables}

\setcounter{figure}{0}
\renewcommand{\thefigure}{A.\arabic{figure}}

\begin{figure}[ht]
    \centering
    \captionsetup{font=footnotesize}
    \includegraphics[width=1.0\textwidth]{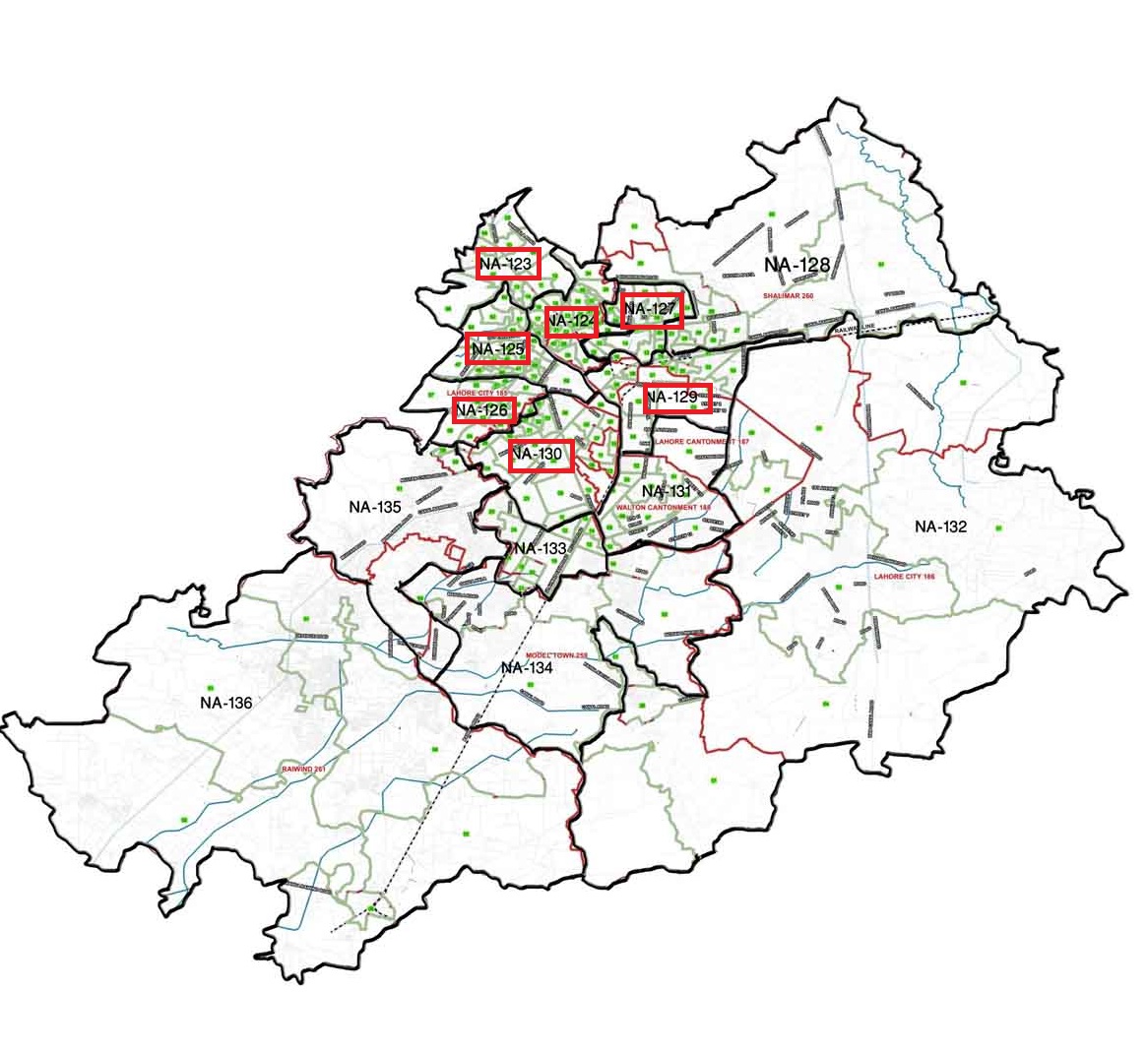}
    \caption{Map of Sample Areas in Lahore}
% \caption*{The map shows the city of Lahore. The constituencies from which the sample is drawn are highlighted in red.}
    %\begin{figurenotes}
    \caption*{\footnotesize{The map obtained from the Pakistan Bureau of Statistics (PBS) shows the National Assembly constituencies in the city of Lahore. The constituencies from which the sample is drawn are highlighted in red.}}
    %\end{figurenotes}
    \label{fig:samp1}
\end{figure}

\begin{figure}[ht]
    \centering
    \captionsetup{font=footnotesize}
    \includegraphics[width=1.0\textwidth]{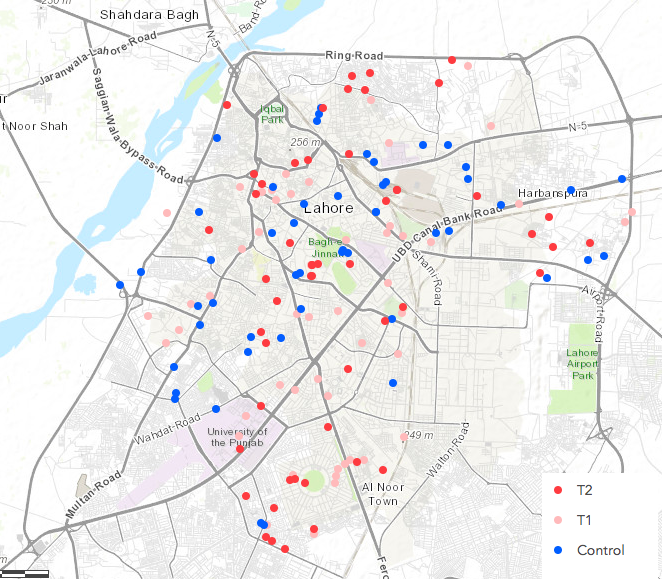}
    \caption{Distribution of Sample Grid Points}
    %\begin{figurenotes}
    \caption*{\footnotesize{The map shows the grid points (geographic clusters) included in the RCT sample, coded by treatment status.}}
    %\end{figurenotes}
   % \caption*{The map shows the grid points included in the sample.}
    \label{fig:samp2}
\end{figure}

\begin{figure}[ht]
    \centering
    \captionsetup{font=footnotesize}
    \includegraphics[width=1.0\textwidth]{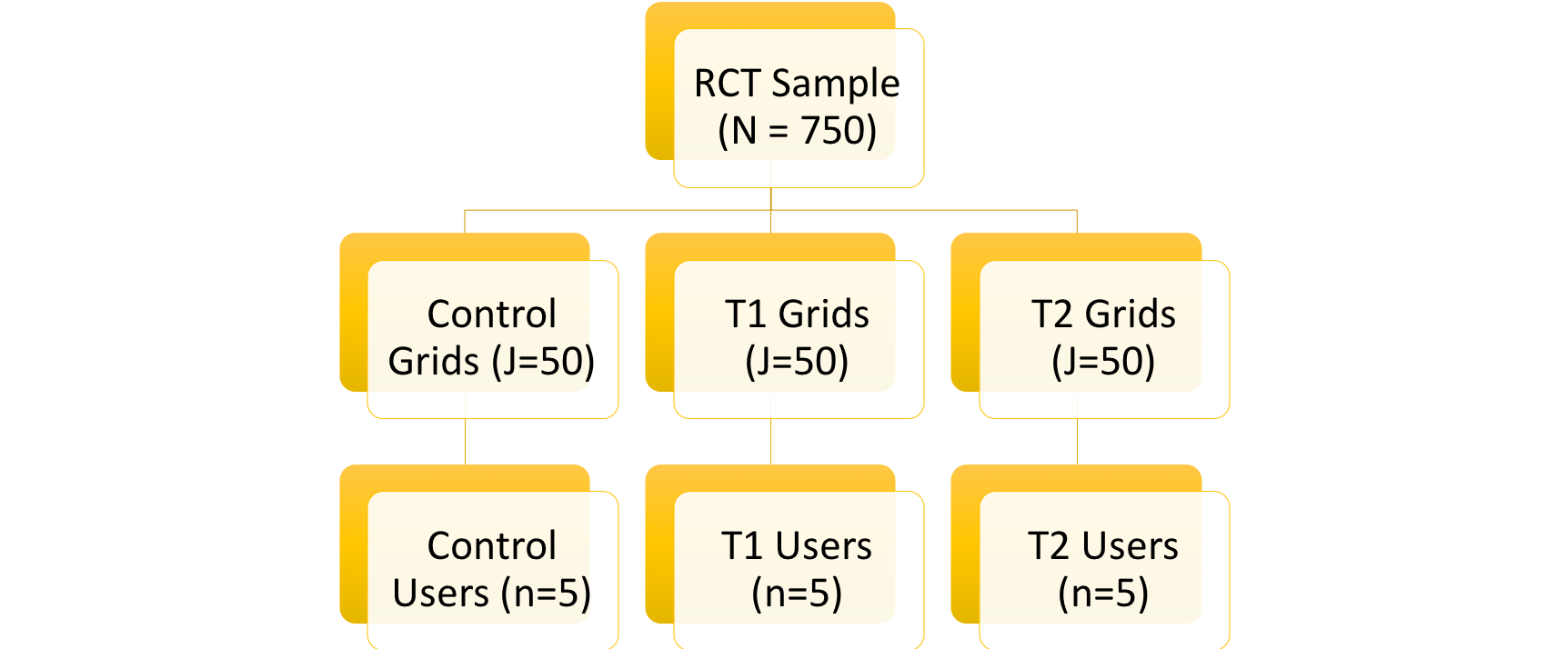}
    \caption{Randomization Scheme}
    \label{fig:randomization}
\end{figure}

\newpage
\newpage

\clearpage
\setcounter{table}{0}
\renewcommand{\thetable}{A.\arabic{table}}

\begin{table}
\centering
\begin{threeparttable}
\footnotesize
\captionsetup{font=footnotesize}
\caption{Survey Response Rates}
\label{table:survey_response_rates}
\begin{tabular}{l*{6}c}
\toprule
& \multicolumn{3}{c}{Means and Standard Deviation} & \multicolumn{3}{c}{Difference in Means and p-values}\\
% \hline
& Control & T1 & T2 & T1 - Control & T2 - Control& T1 - T2\\
\toprule
Endline 1 &0.968&0.968&0.984&0.000&0.016&-0.016\\
&(0.176)&(0.176)&(0.126)&(1.000)&(0.231)&(0.227)\\
Endline 2 &0.980&0.996&1.000&0.016*&0.02**0&-0.004\\
&(0.140)&(0.063)&(0.000)&(0.091)&(0.021)&(0.316)\\
COVID-19 followup&0.825&0.829&0.827&0.004&0.003&0.001\\
&(0.381)&(0.378)&(0.379)&(0.907)&(0.946)&(0.969)\\
Shared mobile number at baseline&0.944&0.936&0.924&-0.008&-0.020&0.012\\
&(0.231)&(0.245)&(0.266)&(0.729)&(0.420)&(0.614)\\
Observations & 251 & 251 & 249 & 502 & 500 & 500 \\
\bottomrule
    \end{tabular}

    \label{tab:attrition}
\begin{tablenotes}
      \footnotesize
      \item Notes: Mean and standard deviation (in parentheses) of the response rates by endline and treatment status is shown in columns 1-3.  Difference in means and p-value of the difference (in parentheses) are shown in columns 4-6. Standard errors are clustered by sample grid. * p$<$0.1, ** p$<$0.05, *** p$<$0.01.
\end{tablenotes}
\end{threeparttable}
\end{table}

\begin{ThreePartTable}
\begin{TableNotes}
\footnotesize
\item Notes: Columns 1 to 3 show the mean and standard deviation of the variable by control, treatment 1, and treatment 2. Columns 4 to 6 report the p-value of the difference in means of the three groups compared pairwise. Education levels are measured on a categorical scale from 1-9 corresponding to no education, primary, middle, matric (grade 10), intermediate (grade 2), bachelors, masters, above masters, diploma. WhatsApp (Facebook) features aware and features used is the the fraction of features that the respondent is aware of and fraction of features used. Baseline false recall is equal to 1 if the respondent says yes to recalling placebo (constructed) news item.
\end{TableNotes}
\footnotesize

\begin{longtable}[c]{lp{1.2cm}p{1.2cm}p{1.2cm}p{1.2cm}p{1.2cm}p{1.2cm}}
\captionsetup{font=footnotesize}
 \caption[c]{Balance Checks\label{balance}}\\
 \toprule
& \multicolumn{3}{c}{Means and Standard Deviation}& \multicolumn{3}{c}{Difference in Means p-value}\\
 \hline
& Control & T1 & T2 & T1 vs. Control & T2 vs. Control& T1 vs. T2\\
 \hline
 \endfirsthead

 %\hline
 \multicolumn{7}{l}{\textit{Continuation of Table \ref{balance}}}\\
 \hline
& \multicolumn{3}{c}{Means and Standard Deviation}& \multicolumn{3}{c}{Difference in Means p-values}\\
 \hline
%& Control & Treatment 1 & Treatment 2 & Treatment 1 vs. Control & Treatment 2 vs. Control\\
%& (1) & (2) & (3) & (4) & (5) & (6)\\
& Control & T1 & T2 & T1 vs. Control & T2 vs. Control& T1 vs. T2\\

 \hline
 \endhead
 
 \hline
 \endfoot
 
%\hline
 %\multicolumn{6}{l}{\textit{End of Table}}\\
 \hline
 \insertTableNotes
 \endlastfoot
 
 \multicolumn{7}{l}{Panel A: Household Characteristics}\\
 \hline
Monthly expenditure & 36.450&33.346&34.949&0.345&0.667&0.630\\
(Rs. '000s)&(20.077)&(26.997)&(27.312)&&&\\

%Household Size & 6.303&6.080&6.655&0.312&0.169&0.019\\
%&(2.376)&(2.440)&(2.840)&&&\\
Fraction of social media users & 0.448&0.452&0.454&0.871&0.789&0.926\\
&(0.225)&(0.210)&(0.222)&&&\\

 \hline 
 \multicolumn{7}{l}{Panel B: Individual Characteristics}\\
 \hline
 Female &0.522&0.498&0.506&0.347&0.538&0.714\\
&(0.501)&(0.501)&(0.501)&&&\\

 Age & 30.143&28.992&29.394&0.249&0.456&0.631\\
&(10.752)&(9.607)&(9.508)&&&\\

 Education level & 4.817&4.865&4.739&0.746&0.613&0.386\\
&(1.382)&(1.286)&(1.289)&&&\\

%Have Own Phone & 1.000& 0.964& 0.972& 0.002& 0.007\\
 %&  (0.000)& (0.186)& (0.166)& & \\
 
 Own smart phone & 0.976&0.920&0.924&0.012&0.014&0.897\\
&(0.153)&(0.271)&(0.266)&&&\\

Have WhatsApp account &0.968&0.952&0.964&0.374&0.803&0.477\\
&(0.176)&(0.214)&(0.187)&&&\\

Have Facebook account & 0.470&0.458&0.470&0.804&0.996&0.812\\
&(0.500)&(0.499)&(0.500)&&&\\

Have Twitter Account & 0.016&0.016&0.004&1.000&0.172&0.172\\
&(0.125)&(0.125)&(0.063)&&&\\ \\

Hours on social media&	3.104&3.137&3.158&0.916&0.862&0.941\\
&(2.680)&(2.741)&(2.373)&&&\\ 
%TV as Primary Source of News&	0.869&	0.888&	0.867&	0.496&	0.972\\
%	&(0.339)&	(0.315)&	(0.34)	&&\\	
Social media as primary	&0.092&0.096&0.100&0.899&0.780&0.886\\
news source&(0.289)&(0.295)&(0.301)&&&\\	

%Share News on Social Media&	0.378&0.466&0.386&0.195&0.916&0.243\\
%&(0.486)&(0.500)&(0.488)&&&\\

%Share News on Facebook&	0.117&	0.112&	0.127&	0.791&	0.614\\
%&	(0.216)&	(0.205)&	(0.226)	&&\\	
%Share News on WhatsApp	&0.166&	0.195&	0.171&	0.209&	0.849\\
%&(0.253) &(0.261)& (0.256)&&\\

%Number of Friends on Facebook&	505.729&	366.718&	447.756&	0.173&	0.570\\
%&	(839.382)&	(713.747)&	(727.101)&&\\		
%Member of WhatsApp Group&	0.271&	0.211&	0.269&	0.118&	0.963\\
%&	(0.445)&	(0.409)&	(0.444)	&&\\	
%Average Size of WhatsApp Group&	6.177&	3.530&	6.124&	0.130&	0.979\\
%&	(22.615)&	(15.889)&	(21.073)&&\\

View content on social media&	1.000&0.992&0.996&0.154&0.318&0.565\\
&(0.000)&(0.089)&(0.063)&&&\\

Share content on social media&	0.669&0.677&0.715&0.897&0.406&0.517\\
&(0.471)&(0.468)&(0.452)&&&\\

Create content on social media&	0.410&0.367&0.438&0.511&0.703&0.288\\
&(0.493)&(0.483)&(0.497)&&&\\

Connect to WiFi and/or&	0.825&0.841&0.815&0.705&0.819&0.543\\
mobile data&(0.381)&(0.367)&(0.389)&&&\\

Use Google search&0.801&0.805&0.743&0.931&0.206&0.177\\
&(0.400)&(0.397)&(0.438)&&&\\

Use social media without&0.861&0.869&0.819&0.840&0.278&0.212\\
assistance&(0.347)&(0.339)&(0.386)&&&\\
Read English on social media &	0.649&0.598&0.643&0.281&0.889&0.346\\
&(0.478)&(0.491)&(0.480)&&&\\

WhatsApp features aware &0.911&0.870&0.861&0.059&0.045&0.756\\
&(0.142)&(0.185)&(0.209)&&&\\
WhatsApp features used &0.855&0.838&0.820&0.443&0.166&0.509\\
&(0.194)&(0.204)&(0.241)&&&\\
Facebook features aware&0.802&0.756&0.741&0.193&0.111&0.714\\
&(0.215)&(0.222)&(0.239)&&&\\
Facebook features used&0.764&0.713&0.698&0.104&0.071&0.698\\
&(0.208)&(0.229)&(0.261)&&&\\

Ever believed fake news &0.410&0.474&0.447&0.293&0.533&0.687 \\
&(0.492)&(0.499)&(0.497)&&&\\

Baseline news score &0.637&0.639&0.655&0.947&0.405&0.443\\
&(0.158)&(0.170)&(0.193)&&&\\

Baseline false recall &0.072&0.100&0.124&0.273&0.071&0.429\\
&(0.259)&(0.300)&(0.331)&&&\\

\bottomrule
\end{longtable}
\end{ThreePartTable}

\begin{table}[t]
\centering
\begin{threeparttable}
\footnotesize
%\hline
\captionsetup{font=footnotesize}
\caption{Treatment Effect on Placebo News }
\label{table:effect_on_placebo}
\begin{tabular}{l p{1.5cm} p{1.5cm}p{1.5cm}p{1.5cm}p{1.5cm}}
\toprule

 \multicolumn{6}{l}{Panel A - Dependent variable: Correctly identify news} \\
 \hline
Treatment 1 &-0.034 &&&&\\
& (0.048) &&&&\\
Treatment 2 & -0.040&&&&\\
& (0.043) &&&&\\
\hline
 Panel B - Dependent variable:& Emotional&Do nothing&Discuss with friends/family & Search online & Share \\
 \hline
 Treatment 1 &-0.002 & -0.048 & 0.001 & 0.094 & 0.013\\
 & (0.081) & (0.104) & (0.086) & (0.091) & (0.036)\\
% P-value& 0.981 & 0.645 & 0.990 & 0.306 & 0.714 \\
% Q-value & 1.000&1.000&1.000&1.000&1.000 \\
 Treatment 2 & -0.050 & 0.091 & -0.139* & 0.009 & -0.014\\
 & (0.084) & (0.089) & (0.075) & (0.082) & (0.034) \\
% P-value & 0.551 & 0.308 & 0.066 & 0.908 & 0.685\\
% Q-value & 1.000&1.000&1.000&1.000&1.000 \\
 
 \hline
 Panel C - Dependent variable:& Prior Belief&Source&Quality & Bias&\\
\hline
Treatment 1&-0.050 & -0.050 & 0.121 & -0.030 &\\
&	(0.081) & (0.084) & (0.106) & (0.068)&\\
%P-value& 0.538 & 0.551 & 0.254 & 0.664&\\
% Q-value & 1.000&1.000&1.000&1.000&1.000 \\

Treatment 2&-0.085 & 0.023 & 0.124 & -0.046 & \\
&(0.080) & (0.084) & (0.104) & (0.061) &\\
%P-value& 0.286 & 0.785 & 0.232 & 0.451&\\
% Q-value & 1.000&1.000&1.000&1.000&1.000 \\
 
\bottomrule
\end{tabular}
\begin{tablenotes}
    \footnotesize
      \item Notes: Sample is only placebo news items shown in first endline. All regressions include baseline controls. Standard errors clustered by sample grid in parentheses.* p$<$0.1,  ** p$<$0.05,*** p$<$0.01.
\end{tablenotes}
\end{threeparttable}
%\vspace{-0.1in}
\end{table}
%\end{comment}

\begin{table}[t]
\centering
\begin{threeparttable}
\footnotesize
%\hline
\captionsetup{font=footnotesize}
\caption{Treatment Effect by Digital Literacy Variables}
\label{table_whatsapp_fb}
\begin{tabular}{lp{1.5cm} p{1.5cm}p{1.5cm}p{1.5cm}p{1.5cm}p{1.5cm}}
\toprule
\multicolumn{4}{l}{Dependent Variable: Correctly identify news}\\

&Basic Score&WhatsApp Score&Facebook Score\\
\toprule

\multicolumn{4}{l}{\textit{Panel A: Fake News}}\\
\\
%\hline
Treatment 1&-0.026&	-0.021&	0.024\\
&(0.061)&	(0.062)	&(0.078)\\
Treatment 1 x Score & 0.103***& 0.094***&0.020\\
& (0.038) & (0.033) & (0.056)	\\
Treatment 1 High digital literacy & 0.0768&	0.0733	&0.0437\\
& (0.079) &	(0.080) &	(0.105)\\
Treatment 2 &	0.146**	&0.142**&	0.231***\\
&(0.057)	&(0.058)&	(0.076)\\
Treatment 2 x Score & 0.079**&  0.065**& -0.024\\
& (0.033) & (0.033)&(0.050)	\\
Treatment 2 High digital literacy & 0.226*** &	0.206*** &	0.207**\\
& (0.070) &	(0.076) & 	(0.101)\\
Observations & 8,111& 	7,798 &	3,765 \\
\hline
\multicolumn{4}{l}{\textit{Panel B - True News}}\\
\\

Treatment 1&	-0.007	&-0.005	&-0.020\\
&(0.037)&	(0.038)	&(0.046)\\
Treatment 1 x Score & -0.015& 0.022 & 0.052\\
& (0.034) &(0.031) & (0.049)\\
Treatment 1 High digital literacy & -0.0219	&0.0168 &	0.0316 \\
& (0.050) &	(0.050) &	(0.068)\\

Treatment 2	&-0.034&	-0.044&	0.003\\
&(0.037)	&(0.038)&	(0.049)\\
Treatment 2 x Score & 0.004 & 0.004 & 0.043\\
& (0.036) & (0.033)&(0.043)\\

Treatment 2 High digital literacy & -0.030& 	-0.039 &	0.047 \\
& (0.053) & 	(0.052) &	(0.071)\\
Observations & 4,428& 	4,257 &	2,055 \\

\bottomrule
\end{tabular}
\begin{tablenotes}
      \footnotesize
      \item Notes: Sample is all news items shown in the two endlines. Standard errors clustered by sample grid in parentheses. See text for details about construction of Basic, Whatsapp and Facebook scores. Treatment effect on high digital literacy group is the total effect of treatments one standard deviation above the mean score. Out of a total of 750 individuals, the number of individuals who use WhatsApp and Facebook are 720 and 348 respectively.* p$<$0.1, ** p$<$0.05, *** p$<$0.01.
\end{tablenotes}
\end{threeparttable}
%\vspace{-0.1in}
\end{table}

\begin{table}
\centering
\begin{threeparttable}
\footnotesize
%\hline
\captionsetup{font=footnotesize}
\caption{Treatment Effect by Demographic Variables}
\label{table:hetero_by_demo}
\begin{tabular}{l lll}
\toprule
\multicolumn{3}{l}{Dependent variable: Correctly identify news}\\

&Gender&Age & Education\\
\toprule

\multicolumn{4}{l}{\textit{Panel A: Fake News}}\\
\\
Treatment 1&	0.032&	-0.004&	-0.060\\
&	(0.077)	&(0.073)&	(0.059)	\\
Treatment 1 x Moderator & -0.099 & -0.027 & 0.146*\\
& (0.078)& (0.068)& (0.081) \\
Treatment 1 effect on subgroup & -0.066 & -0.032 & 0.086 \\
& (0.069) & (0.068) & (0.093) \\
Treatment 2	&0.184**&	0.165**	&0.113*\\
	&(0.073)&	(0.068)&	(0.060)	\\
Treatment 2 x Moderator & -0.074 & -0.036& 0.100\\
& (0.078)& (0.068)& (0.080) \\
Treatment 2 effect on subgroup & 0.110* & 0.129 & 0.213**\\
& (0.067) & (0.067) & (0.085)\\
	
\hline
\multicolumn{4}{l}{\textit{Panel B: True News}}\\
\\
Treatment 1&	0.018&	-0.003&	0.002\\
&	(0.077)	&(0.059)&	(0.046)	\\
Treatment 1 x Moderator &-0.071 &-0.004 & -0.029\\
& (0.071) & (0.074)&(0.067)\\
Treatment 1 effect on subgroup  & -0.054&-0.037 &-0.027 \\
& (0.056) & (0.054) & (0.062)\\
Treatment 2	&0.030&	-0.049 &-0.033\\
	&(0.054)&	(0.056)&	(0.042)	\\
Treatment 2 x Moderator & -0.159* & -0.004& -0.046\\
& (0.081) & (0.074)& (0.075)\\
Treatment 2 effect on subgroup  & -0.129** & -0.053&-0.079\\
& (0.058)& (0.052)&(0.069)\\
	\bottomrule
\end{tabular}
\begin{tablenotes}
      \footnotesize
      \item Notes: Sample is all news items shown in the two endlines. Standard errors clustered by sample grid in parentheses. Moderators are indicator for female, above median age (27 years), and above median education (Intermediate/Grade 12). The total treatment effect for subgroups female, above median age and above median education is reported below the interaction effects. * p$<$0.1, ** p$<$0.05, *** p$<$0.01.
\end{tablenotes}
\end{threeparttable}
\end{table}

%\clearpage
\begin{table}
\centering
\begin{threeparttable}
\footnotesize
\captionsetup{font=footnotesize}
\caption{Treatment Effects Heterogeneity for Females by Digital Literacy}
\label{table:hetero_for_females_by_dl}
\begin{tabular}{l p{1.5cm} p{1.5cm} p{1.5cm}}
\toprule
\multicolumn{4}{l}{Dependent variable: Correctly identify news} \\

& Fake&True &All\\
& (1) & (2) & (3) \\
\toprule
Treatment 1	&	-0.079 & -0.073&-0.077\\
&	(0.069)	 &(0.055) & (0.054)\\
Treatment 1x digital literacy score	&	0.121*** & -0.017&0.073*	\\
&(0.045)& (0.052) & (0.030)	\\
Treatment 1 on high digital literacy&0.043	&-0.091 & -0.004\\
&(0.088)&	(0.070)& (0.067) \\
Treatment 2	&	0.118*	&	 -0.136**& 0.028\\
		&	(0.066)	&	(0.059) & (0.053) \\
Treatment 2 x digital literacy score &	0.080*	&	 -0.057& 0.031\\
&	(0.043)	&	(0.054) & (0.037) \\
Treatment 2  on high digital literacy 	&0.198**& -0.194**&0.056	\\
&(0.086)& (0.082) &(0.071)	\\
Observations&4,167& 2,274&6,441\\
\bottomrule
\end{tabular}
\begin{tablenotes}
     \footnotesize
      \item Notes: Sample is only females and news items shown in the two endlines. All regressions include standard controls. Digital literacy score is standardized by the mean and standard deviation of the control group. The total effect of each treatment on the subgroup is reported. Standard errors clustered by sample grid in parentheses. 
\end{tablenotes}
\end{threeparttable}
\end{table}

\begin{table}[t]
\centering
\begin{threeparttable}
\footnotesize
%\hline
\captionsetup{font=footnotesize}
\caption{Treatment Effect by Gender and Social Desirability Trait }
\label{table:effect_by_social_gender}
\begin{tabular}{l p{1.5cm} p{1.5cm}p{1.5cm}}
\toprule
\multicolumn{4}{l}{Dependent variable: Correctly identify news}\\
 & Fake & True &  All \\
\toprule
\multicolumn{4}{l}{\textit{Panel A - Females}}\\

Treatment 1 & -0.028 & -0.061 & -0.040	\\
 & (0.084) & (0.075) & (0.067) \\
 Treatment 1 x Social Desirability & -0.079 & -0.015 & -0.056\\
& (0.113) & (0.099) & (0.088)\\

Treatment 2	 & 0.138 & -0.124 & 0.046  \\
& (0.093) & (0.084) & (0.076)\\
Treatment 2 x Social Desirability  &  -0.012 & 0.009 & -0.004\\
& (0.121) & (0.112) & (0.096)  \\ 

\hline
\multicolumn{4}{l}{\textit{Panel B - Males}}\\

Treatment 1 & -0.011 & -0.008 & -0.010	\\
 & (0.094) & (0.065) & (0.074) \\
 Treatment 1 x Social Desirability & 0.071 & 0.044 & 0.043\\
& (0.130) & (0.105) & (0.099)\\

Treatment 2	 & 0.153* & 0.008 & 0.102  \\
& (0.089) & (0.070) & (0.070)\\
Treatment 2 x Social Desirability  &  0.070 & 0.044 & 0.061\\
& (0.137) & (0.105) & (0.104)  \\ 

\bottomrule
\end{tabular}
\begin{tablenotes}
    \footnotesize
      \item Notes: All regressions include baseline controls and endline dummy. Standard errors clustered by sample grid in parentheses.* p$<$0.1,  ** p$<$0.05,*** p$<$0.01.
\end{tablenotes}
\end{threeparttable}
%\vspace{-0.1in}
\end{table}

\begin{table}
\centering
\begin{threeparttable}
\footnotesize
\captionsetup{font=footnotesize}
\caption{Effect of Treatment on News Recall}
\label{table:te_recall}
\begin{tabular}{l llll}
\toprule
\multicolumn{5}{l}{Dependent Variable: Recall seeing news before} \\
\hline
& Covid (Fake)& Fake& True& All \\
\cline{2-5}

Treatment 1	&-0.005&	-0.004	&	-0.025&-0.011\\
		&(0.013)&	(0.012))	&	(0.027)&	(0.013) \\
Treatment 2	&0.009&	-0.006	&	-0.047*&-0.02 \\
		&(0.013)&	(0.011)	&	(0.026)&	(0.013) \\

Observations & 4,347& 8,111& 4,428& 12,539 \\
\bottomrule
\end{tabular}
\begin{tablenotes}
     \item Notes: Dependent variable is equal to 1 if the individual recalls seeing or hearing the news from before. Column 1 includes all COVID-19 fake news included in the phone survey. COVID-19 true news is omitted as recall is 100\%.. In columns 2-4, the sample is fake,true, and all news items shown post-treatment in the two endline survey. * p$<$0.1, ** p$<$0.05, *** p$<$0.01.
\end{tablenotes}
\end{threeparttable}
\end{table}

\clearpage
\section{Indices}
\label{sec:app_indices}
\begin{table}[ht]
\footnotesize
    \centering
     \caption{Survey questions corresponding to the information behavior index and the preventative behavior index.}
    \label{tab:indices}
    \begin{tabular}{p{5cm}|p{11cm}}
    \hline
      Information behavior index   &  (1) How frequently are you getting information from each of the following sources about the coronavirus through any medium (including reading online, watching on TV, etc.)? [\emph{For each source: The World Health Organization (WHO) and Government Sources (e.g., covid.gov.pk, press briefings, etc); Responses: Often, Sometimes, Rarely, Never}], coded as 0,0, 1/3,2/3 and 1 respectively. \\
      & (2) How much trust and confidence do you have in each of the following sources when it comes to reporting about the coronavirus? [\emph{For each platform: The World Health Organization (WHO) and Government Sources (e.g., covid.gov.pk, press briefings, etc); Responses: A lot, To some extent, Less, Very less}], coded as 1,3/4,2/4 and 1/4 respectively.\\
      & (3) How frequently do you use the following social media platforms for receiving/sharing news about the coronavirus? [\emph{For each platform: WhatsApp, Facebook, Twitter, YouTube; Responses: Often, Sometimes, Rarely, Never}], coded as 1,2/3,1/3, and 0 respectively. \\
      & (4) How much trust and confidence do you have in news about the coronavirus received through the following social media platforms? [\emph{For each platform: WhatsApp, Facebook, Twitter, YouTube; Responses: A lot, To some extent, Less, Very less}], coded as 0,1/4,2/4, and 3/4 respectively. \\
      & (5) In the last week, on average how often did you share information about coronavirus with others on social media? [\emph{Responses: Never, Once, 2-3 times, 4-5 times, More than five times}], coded as 1, 1/2, 1/3.5, 1/5.5 and 1/7.5 respectively. \\ \hline
      
      Preventative behavior index & (1) On average how many people do you interact with outside your house (in close proximity within 2 arms length), everyday? , coded as 1/(1+x) where x is the number of people.\\
      & (2) On a scale of 1-5, to what extent do you agree or disagree with the following statement: `I will send my children back to school tomorrow if they re-open' [\emph{Responses: 1 (Strongly Agree) - 5 (Strongly Disagree)}], coded as 0,0.24, 0.5, 0.75,  and 1 respectively. \\
      & (3) What type of mask do you have? (can select multiple options) [\emph{Options: No mask, Cloth mask, Non-surgical mask (e.g., blue/green colored ones), KN95 mask (or N95 mask)}], coded as 0 if no mask, 0.5 if own 1 type of mask and 1 if own multiple types of masks.\\
      & (4) In your opinion, what is the market price of each of the following types of masks? [\emph{Cloth mask, Non-surgical mask (e.g., blue/green colored ones), KN95 mask (or N95 mask)}] \\ \hline
    \end{tabular}
   
\end{table}

\end{document}